\documentclass[review,hidelinks,onefignum,onetabnum]{siamart220329}

\usepackage{graphicx}
\usepackage{amsmath,amssymb}
\usepackage{amssymb}
\usepackage{amsfonts}
\usepackage{bm}
\usepackage{cite}
\usepackage{algorithmic}
\usepackage{algorithm,algorithmic}
\usepackage{textcomp}

\usepackage{lipsum}
\usepackage{amsfonts}
\usepackage{graphicx}
\usepackage{epstopdf}
\usepackage{algorithmic}
\usepackage{subcaption}

\ifpdf
  \DeclareGraphicsExtensions{.eps,.pdf,.png,.jpg}
\else
  \DeclareGraphicsExtensions{.eps}
\fi

\newsiamremark{remark}{Remark}
\newsiamremark{hypothesis}{Hypothesis}
\crefname{hypothesis}{Hypothesis}{Hypotheses}
\newsiamthm{claim}{Claim}

\newtheorem{example}{Example}
\newcommand{\oomit}[1]{}

\title{A Framework for Safe Probabilistic Invariance Verification of Stochastic Dynamical Systems \thanks{Submitted to the editors DATE.}}

\author{Taoran Wu\thanks{Key Laboratory of System Software (Chinese Academy of Sciences) and State Key Laboratory of Computer Science, Institute of Software, Chinese Academy of Sciences, Beijing, China 
  (\email{\{wutr,xuebai\}@ios.ac.cn}).} 
\and Yiqing Yu\thanks{Department of Information Science, School of Mathematical Sciences, Peking University, China
  (\email{yuyiqing@hsefz.cn}, \email{xbc@math.pku.edu.cn}).}
\and Bican Xia\footnotemark[3]
\and Ji Wang\thanks{State Key Laboratory of High Performance Computing and College of Computer
Science and Technology at National University of Defense Technology, China 
  (\email{wj@nudt.edu.cn}).}
\and Bai Xue\footnotemark[2]}

\usepackage{amsopn}

\ifpdf
\hypersetup{
  pdftitle={A FRAMEWORK FOR SAFE PROBABILISTIC INVARIANCE
VERIFICATION OF STOCHASTIC DYNAMICAL SYSTEMS},
  pdfauthor={TAORAN WU, YIQING YU, BICAN XIA, JI WANG, AND BAI XUE}
}
\fi

\begin{document}

\maketitle

\begin{abstract}
Ensuring safety through set invariance has proven to be a valuable method in various robotics and control applications. This paper introduces a comprehensive framework for the safe probabilistic invariance verification of both discrete- and continuous-time stochastic dynamical systems over an infinite time horizon. The objective is to ascertain the lower and upper bounds of liveness probabilities for a given safe set and set of initial states. The liveness probability signifies the likelihood of the system remaining within the safe set indefinitely, starting from a state in the initial set. To address this problem, we propose optimizations for verifying safe probabilistic invariance in discrete-time and continuous-time stochastic dynamical systems. These optimizations are constructed via either using the Doob’s nonnegative supermartingale inequality-based method or relaxing the equations described in \cite{xue2021reach,xue2023reach}, which can precisely characterize the probability of reaching a target set while avoiding unsafe states. Finally, we demonstrate the effectiveness of these optimizations through several examples using semi-definite programming tools.
\end{abstract}



\section{Introduction}
\label{sec:intro}
The rapid development of modern technology has led to an increase in intelligent autonomous systems. Ensuring the safe operation of these systems is essential, but external disturbances can cause uncertainties, making it necessary to consider their impact on system safety. Safe robust invariance is commonly used to formalize the impact of unknown perturbations and guarantee that a system will remain inside a specific safe set for all time, regardless of external disturbances. Many studies have been published on certifying safe robust invariance over the past few decades \cite{rakovic2005invariant,xue2021robust}. 

However, the inclusion of all disturbances in safe robust invariance verification often leads to excessive conservatism, making it unable to verify many problems in practical applications. On the other hand, many systems possess additional information regarding external disturbances, such as probability distributions \cite{franzle2008stochastic}. In such cases, safe probabilistic invariance provides a valuable complement to safe robust invariance by ensuring that a system will remain within a specified safe set with a certain probability \cite{kushner1967stochastic}. This approach mitigates the inherent conservatism in safe robust invariance verification by allowing for probabilistic violations and has garnered increasing attention in recent literature \cite{lavaei2022automated}.

This paper proposes a framework for safe probabilistic invariance verification in both stochastic discrete-time and continuous-time systems. The objective of safe probabilistic invariance verification is to compute lower and upper bounds on liveness probabilities of the system remaining inside a safe set for all time, starting from a specific initial set. In this framework, several optimizations are proposed for computing these lower and upper bounds. These optimizations are constructed via the adaptation of stochastic barrier certificates in \cite{anand2022k,prajna2007framework}, which are built upon the well-established Doob's nonnegative supermartingale inequality \cite{ville1939etude}, and equations in \cite{xue2021reach,xue2023reach}, which can precisely characterize the probability of reaching a target set while avoiding unsafe states. To demonstrate and compare the performance of these optimizations, several examples are provided based on the semi-definite programming tool.

The contributions of this work are summarized below.
\begin{enumerate}
\item A framework for the safe probabilistic invariance verification in both stochastic discrete-time and continuous-time systems is proposed. The framework not only computes commonly studied lower bounds of the liveness probabilities, but also provides upper bounds. The incorporation of upper bounds enhances our understanding of the system's invariance properties and allows for a more accurate estimation of the liveness probabilities.
\item The framework incorporates several optimizations to address the safe probabilistic invariance verification. These optimizations are derived from the known Doob’s nonnegative supermartingale inequality, as well as the equations proposed in \cite{xue2021reach,xue2023reach} for characterizing the exact probability. It is demonstrated that the optimizations derived from the nonnegative Doob’s supermartingale inequality are equivalent to those derived from the equations when computing lower bounds. However, for computing upper bounds, the optimizations from \cite{xue2021reach,xue2023reach} are more effective.
\item The performance and effectiveness of the proposed optimizations are extensively demonstrated through numerical examples in both discrete- and continuous-time systems, utilizing the semi-definite programming tool.
\end{enumerate}

This paper extends our previous conference work \cite{yu2023safe}, which focused solely on the safe probabilistic invariance verification in stochastic discrete-time systems. In contrast to \cite{yu2023safe}, this work expands the method to encompass continuous-time systems. Furthermore, we demonstrate that the optimizations derived from the nonnegative Doob’s supermartingale inequality are equivalent to those derived from the equations proposed by \cite{xue2021reach,xue2023reach} when computing lower bounds.



\section{Related Work}

This section specifically focuses on works closely related to the topic of this paper, despite a significant body of research on the verification of stochastic (hybrid) systems. For readers interested in a broader survey of the field, we suggest referring to \cite{lavaei2022automated}. 

\subsection{Stochastic Discrete-time Systems}

The problem addressed in this work is related to the computation of probabilistic invariant sets. These sets define a range of states where a system, starting from any point within the set, must remain within a specified region of interest with a certain probability. By computing probabilistic invariant sets, we can demonstrate that the safe invariance property holds for an initial set if it is a subset of the computed probabilistic invariant sets. Previous works, such as \cite{kofman2012probabilistic}, \cite{kofman2016continuous}, and \cite{hewing2018correspondence}, have approximated polyhedral probabilistic invariant sets using Chebyshev's inequality for linear systems with Gaussian noise. However, these methods are limited to computing lower bounds of the liveness probabilities and cannot be applied to compute upper bounds. Furthermore, the above mentioned works have primarily focused on linear systems, whereas this paper considers nonlinear systems.

On the other hand, probabilistic invariance can be evaluated by examining its dual, probabilistic reachability. \cite{abate2008probabilistic,abate2010approximate} investigated the finite-time probabilistic invariance problem for discrete-time stochastic (hybrid) systems via reachability analysis. Meanwhile, \cite{tkachev2011infinite,tkachev2014characterization} studied the infinite-time probabilistic invariance by defining it as a finite-time reach-avoid property in combination with infinite-time invariance around absorbing sets over the state space and provided a lower bound for the infinite-time probabilistic invariance. Recently, \cite{gao2020computing} proposed an algorithm to compute infinite-horizon probabilistic controlled invariant sets which provide a lower bound on the liveness probabilities based on the dynamic program in \cite{abate2008probabilistic}. These sets are computed using  the stochastic backward reachable set from a robust invariant set. One obstacle faced by the aforementioned methods is the computation of absorbing or robust invariant sets, when they do exist.

Another two closely related works to the present one are \cite{anand2022k}, which studies the safety and reachability verification based on non-negative supermartingale-based barrier certificates, and \cite{xue2021reach}, which formulates a set of equations being able to characterize the exact probability of reaching a specified target set while avoiding unsafe states. With an assumption that the evolution space $\mathcal{X}$ is a robust invariant (i.e., $\bm{f}(\bm{x},\bm{d}): \mathcal{X}\times \mathcal{D}\rightarrow \mathcal{X}$), barrier certificates inspired by the ones in \cite{prajna2007framework} and \cite{prajna2007convex} were formulated for safety and reachability verification of stochastic discrete-time systems in \cite{anand2022k}. The first set of optimizations for discrete-time systems in this paper was adapted from them. The assumption that the evolution space $\mathcal{X}$ (i.e., the safe set in this paper) is a robust invariant is abandoned in our method. Instead, an auxiliary switched system with a robust invariant set is borrowed to construct our constraints for addressing the safe probabilistic invariance verification problem. The second set of optimizations proposed in this paper is inspired by the results in \cite{xue2021reach}. A new equation, which is able to characterize the exact liveness probability of staying with the safe set $\mathcal{X}$, is formulated, and optimizations for addressing the safe probabilistic invariance verification problem are constructed via relaxing this equation.  

\subsection{Stochastic Continuous-time Systems}

Based on the Doob's nonnegative supermartingale inequality, barrier certificates were first proposed in \cite{prajna2007framework} to address the safety verification problem in stochastic continuous-time systems. In \cite{prajna2007framework}, the safety verification problem aims to certify the lower bound of the probability that, starting from a given initial set of states, the system never enters a set of unsafe states when it evolves inside a specified state-constrained set. Subsequently, various barrier certificates were developed, each offering different levels of expressiveness and capabilities in handling diverse problems, e.g., inspired by the results in \cite{kushner1967stochastic} and the Doob's nonnegative supermartingale inequality, \cite{steinhardt2012finite,santoyo2019verification,santoyo2021barrier} proposed barrier-like conditions to tackle the finite-time safety verification problem; temporal logic verification of stochastic systems via barrier certificates was  proposed in \cite{jagtap2018temporal}, which is to find a lower bound on the probability that a
complex temporal logic property is satisfied by finite trajectories of the system. Recently, controlled versions, known as control barrier functions, were explored in \cite{wang2021safety,nishimura2022control}. These controlled versions aim to facilitate the synthesis of controllers that guarantee the invariance of a specified safe set over finite or infinite time horizons, with a probability greater than a specified threshold. On the other hand,  an alternate approach, distinct from the aforementioned ones based on the Doob's nonnegative supermartingale inequality, was presented in \cite{xue2023reach}. It proposes a system of equations capable of accurately describing the probability of the system eventually reaching a target set while remaining within a specified safe set before the first target hitting time. By relaxing this system of equations, it becomes possible to establish both a barrier-like condition for lower-bounding the probability and a barrier-like condition for upper-bounding the probability. Recently, this method was extended to lower- and upper-bound finite-time reachability probabilities in \cite{xue2023new}. In this paper, we adapt the equation from \cite{xue2023reach} to precisely characterize the liveness probability and obtain a set of barrier-like conditions by relaxing the adapted equation for both lower- and upper-bounding the liveness probabilities of continuous-time systems.



This paper is structured as follows. In Section \ref{sec:pre}, we formalize stochastic discrete-time systems and associated safe probabilistic invariance verification problems of interest. In Section \ref{SIV} we present two sets of optimizations designed to address  the safe probabilistic invariance verification problem for stochastic discrete-time systems. In a parallel manner to Section \ref{sec:pre}, we extend our formalization to stochastic continuous-time systems in Section \ref{sec:pre_continuous}. Similarly, Section \ref{sec:method_continuous} is dedicated to presenting two sets of optimizations specifically designed for safe probabilistic invariance verification of stochastic continuous-time systems. Examples demonstrating the proposed optimizations are provided in Section \ref{sec:example}, and finally, this paper is concluded in Section \ref{sec:conclusion}. 

Throughout this paper, we refer to several basic notions. For example, $\mathbb{N}$ is the set of nonnegative integers, while $\mathbb{N}_{\leq k}$ is the set of nonnegative integers that are less than or equal to $k$; $\mathbb{R}_{\geq 0}$ is the set of non-negative real numbers. Additionally, we use the notation $\Delta^c$, $\partial \Delta$ and $\overline{\Delta}$ to represent the complement, boundary and closure of a set $\Delta$, respectively; $\mathcal{C}^2(\Delta)$ denotes a set of twice continuously differentiable functions over $\Delta$. Furthermore, $\mathbb{R}[\cdot]$ denotes the ring of polynomials in variables given by the argument; $\sum[\bm{x}]$ denotes the set of sum-of-squares polynomials over variables $\bm{x}$, i.e., 
$\sum[\bm{x}]=\{p\in \mathbb{R}[\bm{x}]\mid p=\sum_{i=1}^k q^2_i(\bm{x}), q_i(\bm{x})\in \mathbb{R}[\bm{x}],i=1,\ldots,k\}$. Finally, we use the indicator function $1_A(\bm{x})$ to denote whether or not $\bm{x}$ is an element of a set $A$. Specifically, if $\bm{x}\in A$, then $1_A(\bm{x}) = 1$, and if $\bm{x}\notin A$, then $1_A(\bm{x}) = 0$.

\section{Preliminaries on Discrete-time systems}
\label{sec:pre}
We begin by introducing the concept of discrete-time systems that are subject to stochastic disturbances, as well as the problem of verifying safe probabilistic invariance. 

\subsection{Problem Statement}


In this section we are examining stochastic discrete-time systems that are described by stochastic difference equations of the form:
\begin{equation}
\label{system}
\begin{cases}
\bm{x}(l+1)=\bm{f}(\bm{x}(l),\bm{d}(l)), \quad\forall l\in \mathbb{N}\\
\bm{x}(0)=\bm{x}_0 \in \mathbb{R}^n
\end{cases}.
\end{equation}
Here, $\bm{x}(\cdot)\colon \mathbb{N} \rightarrow \mathbb{R}^n$ represents the states, and $\bm{d}(\cdot)\colon \mathbb{N}\rightarrow \mathcal{D}$ with $\mathcal{D} \subseteq \mathbb{R}^m$ represents the stochastic disturbances. The random vectors $\bm{d}(0), \bm{d}(1),\ldots$ are independent and identically distributed (i.i.d), and take values in $\mathcal{D}$ with the probability distribution:
\[{\rm Prob}(\bm{d}(l)\in B)=\mathbb{P}(B), \quad\forall l\in \mathbb{N}, \quad\forall B\subseteq \mathcal{D}.\]
In addition, $\mathbb{E}[\cdot]$ is the expectation induced by $\mathbb{P}$.


To prepare for defining the trajectory of system \eqref{system}, we first need to define a disturbance signal. 
\begin{definition}
A disturbance signal $\pi$ is an ordered sequence $\{\bm{d}(i),i\in \mathbb{N}\}$, where $\bm{d}(\cdot)\colon \mathbb{N}\rightarrow \mathcal{D}$.
\end{definition}

The disturbance signal $\pi$ is a stochastic process defined on the canonical sample space $\Omega=\mathcal{D}^{\infty}$ with the probability measure $\mathbb{P}^{\infty}$, and is denoted by $\{\bm{d}(i),i\in \mathbb{N}\}$. We use $\mathbb{E}^{\infty}[\cdot]$ to represent an expectation with respect to the probability measure $\mathbb{P}^{\infty}$.

Given a disturbance signal $\pi$ and an initial state $\bm{x}_0\in \mathbb{R}^n$, a unique trajectory $\bm{\phi}_{\pi}^{\bm{x}_0}(\cdot)\colon\mathbb{N}\rightarrow \mathbb{R}^n$ is induced with $\bm{\phi}_{\pi}^{\bm{x}_0}(0)=\bm{x}_0$. Specifically, we have $\bm{\phi}_{\pi}^{\bm{x}_0}(l+1)=\bm{f}(\bm{\phi}_{\pi}^{\bm{x}_0}(l),\bm{d}(l))$ for all $l\in \mathbb{N}$.


Given a safe set $\mathcal{X}$ and an initial set $\mathcal{X}_0$, where $\mathcal{X}_0 \subseteq \mathcal{X}$, the safe probabilistic invariance verification problem is to determine lower and upper bounds of liveness probabilities of remaining within the safe set $\mathcal{X}$ for system \eqref{system}, starting from $\mathcal{X}_0$. 

\begin{definition}
\label{ravoid}
The safe probabilistic invariance verification for system \eqref{system} is to compute  lower and upper bounds, denoted by $\epsilon_1\in [0,1]$ and $\epsilon_2\in [0,1]$ respectively, for the liveness probabilities that the system, starting from  $\mathcal{X}_0$, will remain inside the safe set $\mathcal{X}$ for all time, i.e., to compute $\epsilon_1$ and $\epsilon_2$ such that  
\begin{equation}
\label{verification}
\epsilon_1\leq \mathbb{P}^{\infty}(
\forall k\in \mathbb{N}. \bm{\phi}_{\pi}^{\bm{x}_0}(k)\in \mathcal{X})\leq \epsilon_2, \forall \bm{x}_0\in \mathcal{X}_0.
\end{equation}
\end{definition}


\begin{remark}
If the set $\mathcal{X}^{c}$ represents the desired (or, more comfortable) states, then $\epsilon_2$ can be used as an upper bound for the probabilities of the system \eqref{system} getting trapped within $\mathcal{X}$. This means that if we can ensure that all the liveness probabilities are below $\epsilon_2$, we can be confident that the system will eventually reach the desired states with high probability.
\end{remark}

We will focus on addressing the safe probabilistic invariance verification problem defined in Definition \ref{ravoid}.

\subsection{Reachability Probability Characterization in \cite{xue2021reach}}
In this subsection, we will recall an equation that was derived for probabilistic reach-avoid analysis of stochastic discrete-time systems. The equation's bounded solution is equivalent to the precise probability of the system entering a specified target set within a finite time while remaining inside a given safe set before the first target is reached. We will adopt this equation to address the safe invariance verification problem outlined in Definition \ref{ravoid}.

\begin{proposition}[Theorem 1, \cite{xue2021reach}]
\label{theorem_reach}
 Given a safe set $\mathcal{X}$ and a target set $\mathcal{X}_r$, where $\mathcal{X}_r\subseteq \mathcal{X}$, if there exist bounded functions $v(\bm{x})\colon \widehat{\mathcal{X}}\rightarrow \mathbb{R}$ and $w(\bm{x})\colon \widehat{\mathcal{X}}\rightarrow \mathbb{R}$ such that for $\bm{x}\in \widehat{\mathcal{X}}$, 
   \begin{equation}
   \label{reach_equation}
   \begin{cases}
       v(\bm{x})=\mathbb{E}^{\infty}[v(\widehat{\bm{\phi}}^{\bm{x}}_{\pi}(1))]\\
       v(\bm{x})=1_{\mathcal{X}_r}(\bm{x})+\mathbb{E}^{\infty}[w(\widehat{\bm{\phi}}^{\bm{x}}_{\pi}(1))]-w(\bm{x})
   \end{cases}
   \end{equation}
   then for $\bm{x}_0\in \mathcal{X}$,
   \begin{equation*}
   \begin{split}
   & \mathbb{P}^{\infty}\big(\exists k\in \mathbb{N}. \bm{\phi}^{\bm{x}_0}_{\pi}(k) \in \mathcal{X}_r\wedge \forall l\in \mathbb{N}_{\leq k}. \bm{\phi}^{\bm{x}_0}_{\pi}(l)\in \mathcal{X} \big)\\
   &=\mathbb{P}^{\infty}\big(\exists k\in \mathbb{N}. \widehat{\bm{\phi}}^{\bm{x}_0}_{\pi}(k) \in \mathcal{X}_r \big)=\lim_{i\rightarrow \infty}\frac{\mathbb{E}^{\infty}[\sum_{j=0}^{i-1} 1_{\widehat{\mathcal{X}}\setminus \mathcal{X}}(\widehat{\bm{\phi}}^{\bm{x}_0}_{\pi}(j))]}{i}=v(\bm{x}_0)
   \end{split},
   \end{equation*}where 
$\widehat{\bm{\phi}}^{\bm{x}_0}_{\pi}(\cdot)\colon \mathbb{N}\rightarrow \mathbb{R}^n$ is the trajectory to the system 
\begin{equation*}
\begin{cases}
\bm{x}(j+1)=1_{\mathcal{X}\setminus \mathcal{X}_r}(\bm{x}(j))\cdot \bm{f}(\bm{x}(j),\bm{d}(j))\\
~~~~~~~~~~+1_{\mathcal{X}_r}(\bm{x}(j))\cdot \bm{x}(j)+1_{\widehat{\mathcal{X}}\setminus \mathcal{X}}(\bm{x}(j))\cdot \bm{x}(j), \forall j\in \mathbb{N}\\
\bm{x}(0)=\bm{x}_0
\end{cases},
\end{equation*}
and  $\widehat{\mathcal{X}}$ is a set satisfying $\widehat{\mathcal{X}}\supset \{\bm{x}\in \mathbb{R}^n\mid \bm{x}=\bm{f}(\bm{x}_0,\bm{d}), \bm{x}_0\in \mathcal{X}, \bm{d}\in \mathcal{D}\}\cup \mathcal{X}$.
\end{proposition}

A sufficient condition for certifying lower bounds of the probabilities $\mathbb{P}^{\infty}\big(\exists k\in \mathbb{N}. \bm{\phi}^{\bm{x}_0}_{\pi}(k) \in \mathcal{X}_r\wedge \forall l\in \mathbb{N}_{\leq k}. \bm{\phi}^{\bm{x}_0}_{\pi}(l)\in \mathcal{X}\big)$ for $\bm{x}_0\in \mathcal{X}_0$ can be derived via relaxing \eqref{reach_equation}. It is obtained by adding a constraint $v(\bm{x}) \geq \epsilon_1, \forall \bm{x}\in \mathcal{X}_0$ into the ones in Corollary 2 in \cite{xue2021reach}.

\begin{corollary}
\label{reachability}
   Given a safe set $\mathcal{X}$, a target set $\mathcal{X}_r$ and an initial set $\mathcal{X}_0$, where $\mathcal{X}_0,\mathcal{X}_r\subseteq \mathcal{X}$, if there exist bounded functions $v(\bm{x})\colon\widehat{\mathcal{X}}\rightarrow \mathbb{R}$ and $w(\bm{x})\colon \widehat{\mathcal{X}}\rightarrow \mathbb{R}$ such that 
   \begin{equation*}
   \label{rb}
   \begin{cases}
       v(\bm{x}) \geq \epsilon_1 & \forall \bm{x}\in \mathcal{X}_0\\
       v(\bm{x})\leq \mathbb{E}^{\infty}[v(\widehat{\bm{\phi}}^{\bm{x}}_{\pi}(1))] & \forall \bm{x}\in \widehat{\mathcal{X}}\\
       v(\bm{x})\leq 1_{\mathcal{X}_r}(\bm{x})+\mathbb{E}^{\infty}[w(\widehat{\bm{\phi}}^{\bm{x}}_{\pi}(1))]-w(\bm{x}) & \forall \bm{x}\in \widehat{\mathcal{X}}
   \end{cases},
   \end{equation*}
   which is equivalent to 
      \begin{equation}
   \label{rb1}
   \begin{cases}
       v(\bm{x})\geq \epsilon_1 &\forall \bm{x}\in \mathcal{X}_0\\ 
       v(\bm{x})\leq \mathbb{E}^{\infty}[v(\bm{\phi}^{\bm{x}}_{\pi}(1))]& \forall \bm{x}\in \mathcal{X}\setminus \mathcal{X}_r\\
       v(\bm{x})\leq \mathbb{E}^{\infty}[w(\bm{\phi}^{\bm{x}}_{\pi}(1))]-w(\bm{x}) & \forall \bm{x}\in \mathcal{X}\setminus \mathcal{X}_r\\
       v(\bm{x})\leq 1 & \forall \bm{x}\in \mathcal{X}_r\\
       v(\bm{x}) \leq 0 & \forall \bm{x} \in \widehat{\mathcal{X}} \setminus \mathcal{X}
   \end{cases},
   \end{equation}
   then $\mathbb{P}^{\infty}\big(\exists k\in \mathbb{N}. \bm{\phi}^{\bm{x}_0}_{\pi}(k) \in \mathcal{X}_r\wedge \forall l\in \mathbb{N}_{\leq k}. \bm{\phi}^{\bm{x}_0}_{\pi}(l)\in \mathcal{X} \big)=\mathbb{P}^{\infty}\big(\exists k\in \mathbb{N}. \widehat{\bm{\phi}}^{\bm{x}_0}_{\pi}(k) \in \mathcal{X}_r  \big)\geq \epsilon_1$ for $\bm{x}_0\in \mathcal{X}_0$. 
\end{corollary}



\section{Safe Probabilistic Invariance Verification for Discrete-time Systems}
\label{SIV}
This section presents two sets of optimizations to addressing the safe probabilistic invariance verification problem in Definition \ref{ravoid}. The first set of optimizations is adapted from stochastic barrier certificates for safety and reachability verification in \cite{anand2022k}. The second set of optimizations is inspired by Proposition \ref{theorem_reach} along with Corollary \ref{reachability}.

Similar to \cite{xue2021reach}, in constructing our optimizations we need an auxiliary system as follows:
\begin{equation}
\label{new_system1}
\begin{cases}
\bm{x}(j+1)=\widehat{\bm{f}}(\bm{x}(j),\bm{d}(j))& \forall j\in \mathbb{N}\\
\bm{x}(0)=\bm{x}_0
\end{cases},
\end{equation}
where $\widehat{\bm{f}}(\bm{x},\bm{d})=1_{\mathcal{X}}(\bm{x})\cdot \bm{f}(\bm{x},\bm{d})+1_{\widehat{\mathcal{X}}\setminus \mathcal{X}}(\bm{x})\cdot \bm{x}$ and $\widehat{\mathcal{X}}$ is a set containing the union of the set $\mathcal{X}$ and all reachable states starting from $\mathcal{X}$ within one step, i.e., 
\begin{equation}
\label{sets}
\widehat{\mathcal{X}}\supset \{\bm{x}\in \mathbb{R}^n\mid \bm{x}=\bm{f}(\bm{x}_0,\bm{d}), \bm{x}_0\in \mathcal{X}, \bm{d}\in \mathcal{D}\}\cup \mathcal{X}.
\end{equation}

Given a disturbance signal $\pi$, we define the trajectory to system \eqref{new_system1} as $\widehat{\bm{\phi}}^{\bm{x}_0}_{\pi}(\cdot)\colon $ $\mathbb{N}\rightarrow \mathbb{R}^n$, where $\widehat{\bm{\phi}}^{\bm{x}_0}_{\pi}(0)=\bm{x}_0$. It is easy to observe that $\widehat{\mathcal{X}}$ is a robust invariant of system \eqref{new_system1} according to  $\widehat{\bm{f}}(\bm{x},\bm{d}) \in \widehat{\mathcal{X}}, \forall (\bm{x},\bm{d}) \in \widehat{\mathcal{X}}\times \mathcal{D}$.

Also, since $\widehat{\bm{\phi}}_{\pi}^{\bm{x}_0}(1)=\bm{\phi}_{\pi}^{\bm{x}_0}(1)$ for $\bm{x}_0\in \mathcal{X}$, we have that 
\begin{equation*}
\begin{split}
    &\mathbb{P}^{\infty}(\exists k\in \mathbb{N}. \bm{\phi}^{\bm{x}_0}_{\pi}(k) \in \widehat{\mathcal{X}}\setminus \mathcal{X} \wedge \forall i\in \mathbb{N}_{\leq k-1}. \bm{\phi}^{\bm{x}_0}_{\pi}(i) \in \mathcal{X})\\
    =&\mathbb{P}^{\infty}(\exists k\in \mathbb{N}. \widehat{\bm{\phi}}^{\bm{x}_0}_{\pi}(k) \in \widehat{\mathcal{X}}\setminus \mathcal{X} \wedge \forall i\in \mathbb{N}_{\leq k-1}. \widehat{\bm{\phi}}^{\bm{x}_0}_{\pi}(i) \in \mathcal{X})
    \end{split}
\end{equation*}  and \[\mathbb{P}^{\infty}(\forall k\in \mathbb{N}. \bm{\phi}^{\bm{x}_0}_{\pi}(k) \in  \mathcal{X})=\mathbb{P}^{\infty}(\forall k\in \mathbb{N}. \widehat{\bm{\phi}}^{\bm{x}_0}_{\pi}(k) \in  \mathcal{X}).\]


Given a disturbance signal $\pi$ and an initial state $\bm{x}_0\in \mathcal{X}$, the resulting trajectory $\widehat{\bm{\phi}}^{\bm{x}_0}_{\pi}(\cdot)\colon \mathbb{N}\rightarrow \mathbb{R}^n$ either enters the unsafe set $\widehat{\mathcal{X}}\setminus \mathcal{X}$ in finite time (i.e., $\exists k\in \mathbb{N}. \widehat{\bm{\phi}}^{\bm{x}_0}_{\pi}(k) \in \widehat{\mathcal{X}}\setminus \mathcal{X} \wedge \forall i\in \mathbb{N}_{\leq k-1}. \widehat{\bm{\phi}}^{\bm{x}_0}_{\pi}(i) \in \mathcal{X}$ ) or stays inside the safe set $\mathcal{X}$ always (i.e., $\forall k\in \mathbb{N}. \widehat{\bm{\phi}}^{\bm{x}_0}_{\pi}(k) \in  \mathcal{X}$). Thus, $\mathbb{P}^{\infty}(\exists k\in \mathbb{N}. \widehat{\bm{\phi}}^{\bm{x}_0}_{\pi}(k) \in \widehat{\mathcal{X}}\setminus \mathcal{X} \wedge \forall i\in \mathbb{N}_{\leq k-1}. \widehat{\bm{\phi}}^{\bm{x}_0}_{\pi}(i) \in \mathcal{X})+\mathbb{P}^{\infty}(\forall k\in \mathbb{N}. \widehat{\bm{\phi}}^{\bm{x}_0}_{\pi}(k) \in  \mathcal{X})=1$. Therefore, if the upper and lower bounds of the probability $\mathbb{P}^{\infty}(\exists k\in \mathbb{N}. \widehat{\bm{\phi}}^{\bm{x}_0}_{\pi}(k) \in \widehat{\mathcal{X}}\setminus \mathcal{X} \wedge \forall i\in \mathbb{N}_{\leq k-1}. \widehat{\bm{\phi}}^{\bm{x}_0}_{\pi}(i) \in \mathcal{X})$ are gained, one can obtain the lower and upper bounds of the liveness probability of staying inside the safe set $\mathcal{X}$.



\subsection{Doob's Nonnegative Supermartingale Inequality Based Invariance Verification}

In this subsection we propose optimizations for certifying lower and upper bounds of the liveness probabilities by adapting barrier certificates in \cite{anand2022k} for safety and reachability verification, which are respectively built upon the well established Doob's nonnegative supermartingale inequality \cite{ville1939etude}.  

Proposition \ref{direct_lower} provides a straightforward sufficient condition for lower bounds on the liveness probabilitiess, as demonstrated in Theorem 5 of \cite{anand2022k}.
\begin{proposition}
\label{direct_lower}
    Under the assumption that $\Omega\subseteq \mathbb{R}^n$ is a robust invariant set for system \eqref{system}, i.e., $\bm{f}(\bm{x},\bm{d})\colon \Omega\times \mathcal{D}\rightarrow \Omega$, and $\mathcal{X}\subseteq \Omega$, if there exists $v(\bm{x})\colon \Omega \rightarrow \mathbb{R}_{\geq 0}$ such that 
    \begin{equation}
    \label{direct_lower_c}
        \begin{cases}
            v(\bm{x})\leq 1-\epsilon_1 & \forall \bm{x}\in \mathcal{X}_0\\
            v(\bm{x})\geq 1 & \forall \bm{x}\in \mathcal{X}_{unsafe}(=\Omega\setminus \mathcal{X})\\
            \mathbb{E}^{\infty}[v(\bm{\phi}_{\pi}^{\bm{x}}(1))]-v(\bm{x})\leq 0 & \forall \bm{x}\in \Omega
        \end{cases},
    \end{equation}
 then $\mathbb{P}^{\infty}\big(\exists k\in \mathbb{N}. \bm{\phi}^{\bm{x}_0}_{\pi}(k) \in \mathcal{X}_{unsafe}\big)\leq 1-\epsilon_1$, $\forall \bm{x}_0\in \mathcal{X}_0$. Thus, $\mathbb{P}^{\infty}\big(\forall k\in \mathbb{N}. \bm{\phi}^{\bm{x}_0}_{\pi}(k) \in \mathcal{X}\big)\geq \epsilon_1$, $\forall \bm{x}_0\in \mathcal{X}_0$ .
\end{proposition}

Similarly, a sufficient condition for determining upper bounds of the liveness probabilities can be obtained straightforwardly from Theorem 16 in \cite{anand2022k}. However, finding a robust invariant set $\Omega$, except for the trivial case of $\Omega = \mathbb{R}^n$, can be challenging and computationally intensive for many systems, if it even exists. On the other hand, when $\Omega = \mathbb{R}^n$ in Proposition \ref{direct_lower}, the resulting constraint \eqref{direct_lower_c} may be too strong, leading to an overly conservative lower bound. We use an example to illustrate this below. 
\begin{example}
\label{illu}
  In this example we consider a computer-based model, which is modified from the reversed-time Van der Pol oscillator based on Euler's method with the time step 0.01:
   \begin{equation}
   \label{vander}
      \begin{cases}
      x(l+1)=x(l)+0.01\Big(-2y(l)\Big)\\
      y(l+1)=y(l)+0.01\Big((0.8+d(l))x(l)+10(x^2(l)-0.21)y(l)\Big)
      \end{cases},
      \end{equation}
      where $d(\cdot)\colon\mathbb{N}\rightarrow \mathcal{D}=[-0.1,0.1]$, $\mathcal{X}=\{\,(x,y)^{\top}\mid h(\bm{x})\leq 0\,\}$ with $h(\bm{x})=x^2+y^2-1$, and $\mathcal{X}_0=\{\,(x,y)^{\top}\mid g(\bm{x})<0\,\}$ with $g(\bm{x})=x^2+y^2-0.01$.  We assume that the probability distribution on $\mathcal{D}$ is the uniform distribution. The lower bound of the liveness probabilities estimated via the Monte Carlo method is 1. 
      
      Via solving \textbf{Op0} (shown in Appendix \ref{sec:appendix_opt}), which is encoded into a semi-definite program \textbf{SDP0} (shown in Appendix \ref{sec:appendix_sdp}) via the sum of squares decomposition for multivariate polynomials,
we obtain a lower bound of the liveness probabilities, which is 2.1368e-07. This is too conservative to be useful in practice.  The resulting semi-definite program is addressed when unknown polynomials of degree 8 are used. $\blacksquare$
\end{example}

In the following we will present weaker sufficient conditions for certifying lower and upper bounds using the switched system \eqref{new_system1}. They are respectively formulated in Theorem \ref{barrier} and \ref{reachability_existing}.

\begin{theorem}
\label{barrier}
Given a safe set $\mathcal{X}$ and an initial set $\mathcal{X}_0$ with $\mathcal{X}_0\subseteq \mathcal{X}$, if there exist a barrier certificate $v(\bm{x})\colon\widehat{\mathcal{X}}\rightarrow \mathbb{R}$ satisfying
   \begin{equation}
   \label{reach_equa_barrier}
   \begin{cases}
       v(\bm{x})\leq 1-\epsilon_1 & \forall \bm{x}\in \mathcal{X}_0\\
       v(\bm{x})\geq \mathbb{E}^{\infty}[v(\bm{\phi}^{\bm{x}}_{\pi}(1))] & \forall \bm{x}\in \mathcal{X}\\
       v(\bm{x})\geq 1 & \forall \bm{x}\in \widehat{\mathcal{X}}\setminus \mathcal{X}\\
       v(\bm{x})\geq 0& \forall \bm{x}\in \widehat{\mathcal{X}}
   \end{cases},
   \end{equation}
   then $\mathbb{P}^{\infty}\big(\exists k\in \mathbb{N}. \bm{\phi}^{\bm{x}_0}_{\pi}(k) \in \widehat{\mathcal{X}}\setminus \mathcal{X}\big)\leq 1-\epsilon_1$, $\forall \bm{x}_0\in \mathcal{X}_0$. Consequently, $\mathbb{P}^{\infty}\big(\forall k\in \mathbb{N}. \bm{\phi}^{\bm{x}_0}_{\pi}(k) \in \mathcal{X}\big)\geq \epsilon_1$, $\forall \bm{x}_0\in \mathcal{X}_0$.
\end{theorem}
\begin{proof}
Constraint \eqref{reach_equa_barrier} is equivalent to the following constraint 
   \begin{equation*}
   \label{reach_equa_barrier1}
   \begin{cases}
       v(\bm{x})\leq 1-\epsilon_1 & \forall \bm{x}\in \mathcal{X}_0,\\
       v(\bm{x})\geq \mathbb{E}^{\infty}[v(\widehat{\bm{\phi}}^{\bm{x}}_{\pi}(1))] & \forall \bm{x}\in \widehat{\mathcal{X}}\\
       v(\bm{x})\geq 1 & \forall \bm{x}\in \mathcal{X}_{unsafe}(=\widehat{\mathcal{X}}\setminus \mathcal{X})\\
       v(\bm{x})\geq 0 & \forall \bm{x}\in \widehat{\mathcal{X}}
   \end{cases}.
   \end{equation*}
Therefore, $v(\bm{x})$ is a classical nonnegative supermartingale based barrier certificate for system \eqref{new_system1} with the invariant set $\widehat{\mathcal{X}}$ and the unsafe set $\widehat{\mathcal{X}}\setminus \mathcal{X}$. Therefore, according to Theorem 5 in \cite{anand2022k}, we have the conclusion that the probability of reaching the unsafe set $\widehat{\mathcal{X}}\setminus \mathcal{X}$ for system \eqref{new_system1} starting from each state in $\mathcal{X}_0$ is less than or equal to $1-\epsilon_1$, i.e., $\mathbb{P}^{\infty}\big(\exists k\in \mathbb{N}. \widehat{\bm{\phi}}^{\bm{x}_0}_{\pi}(k) \in \widehat{\mathcal{X}}\setminus \mathcal{X}\big)\leq 1-\epsilon_1$, $\forall \bm{x}_0\in \mathcal{X}_0$. Therefore, \[\mathbb{P}^{\infty}\big(\forall k\in \mathbb{N}. \widehat{\bm{\phi}}^{\bm{x}_0}_{\pi}(k) \in  \mathcal{X}\big)\geq \epsilon_1, \forall \bm{x}_0\in \mathcal{X}_0.\] Since if $\bm{x}_0\in \mathcal{X}$, $\widehat{\bm{\phi}}^{\bm{x}_0}_{\pi}(1)=\bm{\phi}^{\bm{x}_0}_{\pi}(1)$ holds. Consequently, $\mathbb{P}^{\infty}\big(\forall k\in \mathbb{N}. \bm{\phi}^{\bm{x}_0}_{\pi}(k) \in \mathcal{X}\big)\geq \epsilon_1$, $\forall \bm{x}_0\in \mathcal{X}_0$.
\end{proof}

According to Theorem \ref{barrier}, a lower bound of the liveness probabilities can be computed via solving \textbf{Op1} (shown in Appendix \ref{sec:appendix_opt}).

\begin{example}
   Consider Example \ref{illu} again. By solving \textbf{Op1} using $\widehat{\mathcal{X}}=\{\,(x,y)^{\top} \mid \widehat{h}(\bm{x})\leq 0\,\}$ with $\widehat{h}(\bm{x})=x^2+y^2-2$, we encode it into a semi-definite program \textbf{SDP1}(shown in Appendix) via the sum of squares decomposition for multivariate polynomials. The solution yields a lower bound for the liveness probabilities, which is 0.9465. \textbf{SDP1} is addressed when unknown polynomials of degree 8 are used. $\blacksquare$
\end{example}

\begin{theorem}
\label{reachability_existing}
Assume that $\mathcal{X}_{unsafe}=\widehat{\mathcal{X}}\setminus \mathcal{X}$ and $\mathcal{X}$ is a closed set\footnote{The requirement that $\mathcal{X}$ is closed is reflected in \eqref{super_reach1} in the proof.}. If there exists a function $v(\bm{x})\colon \widehat{\mathcal{X}}\rightarrow \mathbb{R}$ satisfying 
    \begin{equation}
    \label{super_reach}
        \begin{cases}
            v(\bm{x})\leq \epsilon_2 & \forall \bm{x}\in \mathcal{X}_0,\\
            v(\bm{x})\geq 1 & \forall \bm{x} \in \partial \widehat{\mathcal{X}}\setminus \partial \mathcal{X}_{unsafe}\\
            \mathbb{E}^{\infty}[v(\bm{\phi}_{\pi}^{\bm{x}}(1))]-v(\bm{x}) \leq -\delta & \forall \bm{x} \in \mathcal{X}\\
             v(\bm{x})\geq 0 & \forall \bm{x}\in \widehat{\mathcal{X}}
        \end{cases},
    \end{equation}
   where $\delta > 0$ is a user-defined value, then $\mathbb{P}^{\infty}\big(\forall k\in \mathbb{N}. \bm{\phi}^{\bm{x}_0}_{\pi}(k) \in \mathcal{X}\big)\leq \epsilon_2$, $\forall  \bm{x}_0\in \mathcal{X}_0$. 
\end{theorem}
\begin{proof}
    Constraint \eqref{super_reach} is equivalent to the following constraint 
   \begin{equation}
   \label{super_reach1}
   \begin{cases}
           v(\bm{x})\leq \epsilon_2 & \forall \bm{x}\in \mathcal{X}_0,\\
            v(\bm{x})\geq 1 & \forall \bm{x} \in \partial \widehat{\mathcal{X}}\setminus \partial \mathcal{X}_{unsafe}\\
            \mathbb{E}^{\infty}[v(\widehat{\bm{\phi}}_{\pi}^{\bm{x}}(1))]-v(\bm{x}) \leq -\delta & \forall \bm{x} \in \overline{\widehat{\mathcal{X}}\setminus \mathcal{X}_{unsafe}}\\
             v(\bm{x})\geq 0 & \forall \bm{x}\in \widehat{\mathcal{X}}
   \end{cases},
   \end{equation}
   According to Theorem 16 in \cite{anand2022k} and following the proof of Theorem \ref{barrier}, we have the conclusion $\mathbb{P}^{\infty}\big(\forall k\in \mathbb{N}. \bm{\phi}^{\bm{x}_0}_{\pi}(k) \in \mathcal{X}\big)\leq \epsilon_2$, $\forall  \bm{x}_0\in \mathcal{X}_0$.  
\end{proof}

\begin{remark}
\label{remark_bounded}
If $v(\bm{x})$ is bounded over $\widehat{\mathcal{X}}$ in \eqref{super_reach}, constraint \eqref{super_reach} provides strong guarantees of leaving the safe set $\mathcal{X}$ almost surely, i.e., $\mathbb{P}^{\infty}\big(\forall k\in \mathbb{N}. \bm{\phi}^{\bm{x}_0}_{\pi}(k) \in \mathcal{X}\big)=0$, $\forall  \bm{x}_0\in \mathcal{X}_0$. This conclusion is justified as follows.

From $\mathbb{E}^{\infty}[v(\widehat{\bm{\phi}}_{\pi}^{\bm{x}}(1))]-v(\bm{x}) \leq -\delta, \forall \bm{x} \in \overline{\widehat{\mathcal{X}}\setminus \mathcal{X}_{unsafe}}$, where $\mathcal{X}_{unsafe}=\widehat{\mathcal{X}}\setminus \mathcal{X}$, we have 
\[\mathbb{E}^{\infty}[v(\widehat{\bm{\phi}}_{\pi}^{\bm{x}}(1))]-v(\bm{x}) -\delta 1_{\mathcal{X}_{unsafe}}(\bm{x})\leq -\delta, \forall \bm{x}\in \widehat{\mathcal{X}}.\]
Thus, for $\bm{x}\in \widehat{\mathcal{X}}$, we have 
\[\mathbb{E}^{\infty}[v(\widehat{\bm{\phi}}_{\pi}^{\bm{x}}(k))]-v(\bm{x}) -\delta \sum_{i=0}^{k-1}\mathbb{E}^{\infty}[1_{\mathcal{X}_{unsafe}}(\widehat{\bm{\phi}}_{\pi}^{\bm{x}}(i))]\leq -k\delta, \]
which implies 
$\mathbb{P}^{\infty}\big(\exists k\in \mathbb{N}. \widehat{\bm{\phi}}^{\bm{x}}_{\pi}(k) \in \widehat{\mathcal{X}}\setminus \mathcal{X} \big)=\lim_{k\rightarrow \infty}\frac{ \sum_{i=0}^{k-1}\mathbb{E}^{\infty}[1_{\mathcal{X}_{unsafe}}(\widehat{\bm{\phi}}_{\pi}^{\bm{x}}(i))]}{k}$ $ \geq 1$ (according to Lemma 3 in \cite{xue2021reach}) Thus, we have the conclusion. 
\end{remark}

It is worth noting that if $\partial \widehat{\mathcal{X}}\cap \partial \mathcal{X}=\emptyset$, the set $\partial \widehat{\mathcal{X}}\setminus \partial \mathcal{X}_{unsafe}$ in \eqref{super_reach} is empty. As a result, the constraint $v(\bm{x})\geq 1, \forall \bm{x} \in \partial \widehat{\mathcal{X}}\setminus \partial \mathcal{X}_{unsafe}$ becomes redundant and can be removed. Throughout this paper, unless explicitly stated otherwise, we assume that $\partial \widehat{\mathcal{X}}\cap \partial \mathcal{X}=\emptyset$. It is worth mentioning that this assumption is not overly strict and can be easily satisfied by enlarging the set that satisfies \eqref{sets}. The primary role of this assumption is to facilitate solving constraint \eqref{super_reach}. 
Accordingly, based on Theorem \ref{reachability_existing}, we can calculate an upper bound for the liveness probabilities by solving \textbf{Op2} (shown in Appendix \ref{sec:appendix_opt}).

\begin{remark}
Another condition, which is analogous to the one in Theorem \ref{reachability_existing}, was proposed in \cite{chakarov2013probabilistic}. 
\begin{proposition}
   If there exist a function  $v(\bm{x})\colon \mathcal{X}\rightarrow \mathbb{R}_{\geq 0}$ and constant $c>0$ such that 
    \begin{equation}
    \label{almost}
        \begin{cases}
            v(\bm{x})\geq c & \forall \bm{x}\in \mathcal{X}\\
            v(\bm{x})<c & \forall \bm{x}\in \mathcal{X}_{unsafe}(=\widehat{\mathcal{X}}\setminus \mathcal{X})\\
            \mathbb{E}^{\infty}[v(\bm{\phi}^{\bm{x}}_{\pi}(1))]-v(\bm{x})\leq -1 & \forall \bm{x}\in \mathcal{X}
        \end{cases},
    \end{equation}
    then $\mathbb{P}^{\infty}\big(\forall k\in \mathbb{N}. \bm{\phi}^{\bm{x}_0}_{\pi}(k) \in \mathcal{X}\big)=0$, $\forall  \bm{x}_0\in \mathcal{X}_0$. 
\end{proposition}
\begin{proof}
The proof is similar to the one of Theorem \ref{barrier}, and the conclusion is justified from Theorem 19 in \cite{anand2022k}. 
\end{proof}
\end{remark}

\begin{remark}
    It is worth remarking here that we do not extend the k-Inductive Barrier certificates proposed in \cite{anand2022k} for addressing the problem in this paper. A set of sufficient conditions, which is similar to the one in Proposition \ref{direct_lower} can be easily obtained based on k-Inductive barrier certificates. However, the gain of sufficient conditions being analogous to the ones in Theorem \ref{barrier} and \ref{reachability_existing} should be carefully treated and will be considered in the future work.    
\end{remark}

\subsection{Equations Relaxation Based Invariance Verification}
In this subsection, we present the second set of optimizations for addressing the safe invariance verification problem in Definition \ref{ravoid}. We begin by introducing an equation that characterizes the liveness probability of staying the safe set $\mathcal{X}$. This equation is adapted from \eqref{reach_equation}, where the unsafe set $\widehat{\mathcal{X}}\setminus \mathcal{X}$ is regarded as the target set $\mathcal{X}_r$ in \eqref{reach_equation}. We then propose two sufficient conditions for certifying lower and upper bounds of the liveness probabilities by relaxing the derived equation. 


\begin{lemma}
\label{invariance}
Given a safe set $\mathcal{X}$, if there exist a function $v(\bm{x})\colon\widehat{\mathcal{X}}\rightarrow \mathbb{R}$ \footnote{Comparing with Proposition \ref{theorem_reach}, the explicit requirement that $v(\bm{x})$ is bounded is abandoned here. If following the proof of Proposition \ref{theorem_reach} in \cite{xue2021reach}, one can find that this requirement is not necessary.} and a bounded function $w(\bm{x})\colon \widehat{\mathcal{X}}\rightarrow \mathbb{R}$ such that for $\bm{x}\in \widehat{\mathcal{X}}$, 
   \begin{equation}
   \label{reach_equa}
   \begin{cases}
       v(\bm{x})=\mathbb{E}^{\infty}[v(\widehat{\bm{\phi}}^{\bm{x}}_{\pi}(1))]\\
       v(\bm{x})=1_{\widehat{\mathcal{X}}\setminus \mathcal{X}}(\bm{x})+\mathbb{E}^{\infty}[w(\widehat{\bm{\phi}}^{\bm{x}}_{\pi}(1))]-w(\bm{x})
   \end{cases},
   \end{equation}
then for $\bm{x}_0\in \mathcal{X}$,
\begin{equation*}
\begin{split}
&v(\bm{x}_0)=\mathbb{P}^{\infty}\big(\exists k\in \mathbb{N}. \bm{\phi}^{\bm{x}_0}_{\pi}(k) \in \widehat{\mathcal{X}}\setminus \mathcal{X}\big)
=\mathbb{P}^{\infty}\big(\exists k\in \mathbb{N}. \widehat{\bm{\phi}}^{\bm{x}_0}_{\pi}(k) \in \widehat{\mathcal{X}}\setminus \mathcal{X} \big)\\
&=\lim_{i\rightarrow \infty}\frac{\mathbb{E}^{\infty}[\sum_{j=0}^{i-1} 1_{\widehat{\mathcal{X}}\setminus \mathcal{X}}(\widehat{\bm{\phi}}^{\bm{x}_0}_{\pi}(j))]}{i}.
\end{split}
\end{equation*} Thereby, $\mathbb{P}^{\infty}\big(\forall k\in \mathbb{N}. \bm{\phi}^{\bm{x}_0}_{\pi}(k) \in \mathcal{X}\big)=1-v(\bm{x}_0)$, $\forall \bm{x}_0\in \mathcal{X}$.
\end{lemma}
\begin{proof}
The conclusion can be assured by following the proof of Theorem 1 in \cite{xue2021reach}.
\end{proof}


Like Corollary \ref{reachability}, two sufficient conditions can be obtained for certifying lower and upper bounds of the liveness probabilities via directly relaxing equation \eqref{reach_equa}, respectively.

\begin{theorem}
\label{upper_invariance}
Given a safe set $\mathcal{X}$ and an initial set $\mathcal{X}_0$ with $\mathcal{X}_0\subseteq \mathcal{X}$, if there exist a function $v(\bm{x})\colon\widehat{\mathcal{X}}\rightarrow \mathbb{R}$ and a bounded function $w(\bm{x})\colon \widehat{\mathcal{X}}\rightarrow \mathbb{R}$ satisfying
   \begin{equation}
   \label{reach_equa2}
   \begin{cases}
       v(\bm{x})\leq 1-\epsilon_1 & \forall \bm{x}\in \mathcal{X}_0\\
       v(\bm{x})\geq \mathbb{E}^{\infty}[v(\bm{\phi}^{\bm{x}}_{\pi}(1))]& \forall \bm{x}\in \mathcal{X}\\
       v(\bm{x})\geq \mathbb{E}^{\infty}[w(\bm{\phi}^{\bm{x}}_{\pi}(1))]-w(\bm{x}) & \forall \bm{x}\in \mathcal{X}\\
       v(\bm{x})\geq 1 & \forall \bm{x}\in \widehat{\mathcal{X}}\setminus \mathcal{X}
   \end{cases},
   \end{equation}
   then $\mathbb{P}^{\infty}\big(\exists k\in \mathbb{N}. \bm{\phi}^{\bm{x}_0}_{\pi}(k) \in \widehat{\mathcal{X}}\setminus \mathcal{X}\big)\leq v(\bm{x}_0) \leq 1-\epsilon_1$, $\forall \bm{x}_0\in \mathcal{X}_0$.  Consequently, $\mathbb{P}^{\infty}\big(\forall k\in \mathbb{N}. \bm{\phi}^{\bm{x}_0}_{\pi}(k) \in \mathcal{X}\big)\geq \epsilon_1$, $\forall \bm{x}_0\in \mathcal{X}_0$.
\end{theorem}
\begin{proof}
The conclusion can be assured by following the proof of Corollary 2 in \cite{xue2021reach}, with the inequality signs reversed.
\end{proof}

According to Theorem \ref{upper_invariance}, a lower bound of the liveness probabilities can be computed via solving \textbf{Op3} (shown in Appendix \ref{sec:appendix_opt}).

Although \eqref{reach_equa_barrier} and \eqref{reach_equa2} are derived using different methods and have different forms, they are equivalent. 

\begin{proposition}
    \label{coro:discrete}
    Constraints \eqref{reach_equa_barrier} and \eqref{reach_equa2} are equivalent.
\end{proposition}
\begin{proof}
Since \eqref{reach_equa_barrier} is typical form of \eqref{reach_equa2} with $w(\bm{x})=0$ for $\bm{x}\in \overline{\mathcal{X}}$,  we only need to prove that $v(\bm{x}) \geq 0, \forall \bm{x} \in \widehat{\mathcal{X}}$ if $v(\bm{x})$ satisfies \eqref{reach_equa2}.
    
    Assume $v(\bm{x})$ satisfies \eqref{reach_equa2} and there exits $\bm{x}_0 \in \widehat{\mathcal{X}}$ such that $v(\bm{x_0}) < -\delta$, where $\delta >0$. According to $v(\bm{x})\geq \mathbb{E}[v(\widehat{\bm{\phi}}^{\bm{x}}_{\pi}(1))], \forall \bm{x}\in \widehat{\mathcal{X}}$, we have  
    \begin{equation}
        \label{eq:coro1_1}
        \mathbb{E}^\infty[v(\widehat{\bm{\phi}}^{\bm{x}_0}_{\pi}(k))] \leq v(\bm{x}_0) < -\delta, \forall k \in \mathbb{N}.
    \end{equation}
    Also, since $v(\bm{x})\geq 1_{\widehat{\mathcal{X}}\setminus \mathcal{X}}(\bm{x}) +\mathbb{E}[w(\widehat{\bm{\phi}}^{\bm{x}}_{\pi}(1))]-w(\bm{x})$ for $\forall \bm{x} \in \widehat{\mathcal{X}}$, we have $v(\bm{x})\geq \mathbb{E}[w(\widehat{\bm{\phi}}^{\bm{x}}_{\pi}(1))]-w(\bm{x}), \forall \bm{x} \in \widehat{\mathcal{X}}.$
    
    
    

  Thus, we have
    \begin{equation*}
    \begin{split}
       (j+1)v(\bm{x}_0)&\geq \sum_{k=0}^j \mathbb{E}^\infty[v(\widehat{\bm{\phi}}^{\bm{x}_0}_{\pi}(k))] 
        \mathbb{E}^\infty[w(\widehat{\bm{\phi}}^{\bm{x}_0}_{\pi}(j+1))] - w(\bm{x}_0),\forall j \in \mathbb{N}. 
    \end{split}
    \end{equation*}

    According to \eqref{eq:coro1_1}, we have that
    \begin{equation*}
    \begin{split}
       \mathbb{E}^\infty[w(\widehat{\bm{\phi}}^{\bm{x}_0}_{\pi}(j+1)] - w(\bm{x}_0) \leq & \sum_{k=0}^j \mathbb{E}^\infty[v(\widehat{\bm{\phi}}^{\bm{x}_0}_{\pi}(k))] < -\delta (j+1),\ \forall j \in \mathbb{N},
    \end{split}
    \end{equation*}
    implying that 
    \begin{equation}
    \label{eq:coro1_3}
        \frac{\mathbb{E}^\infty[w(\widehat{\bm{\phi}}^{\bm{x}_0}_{\pi}(j+1)] - w(\bm{x}_0)}{j+1} < -\delta, \forall j \in \mathbb{N}.
    \end{equation}
    which contradicts
    \[\lim_{j\rightarrow\infty}\frac{\mathbb{E}^\infty[w(\widehat{\bm{\phi}}^{\bm{x}_0}_{\pi}(j+1)] - w(\bm{x}_0)}{j} = 0.\]
     Therefore, we can conclude that if $v(\bm{x})$ satisfies \eqref{reach_equa2}, $v(\bm{x}) \geq 0$ holds for $\bm{x} \in \widehat{\mathcal{X}}$. Thus, conditions \eqref{reach_equa_barrier} and \eqref{reach_equa2} are equivalent.
\end{proof}

\begin{theorem}
\label{lower_invariance}
Given a safe set $\mathcal{X}$ and an initial set $\mathcal{X}_0$ with $\mathcal{X}_0\subseteq \mathcal{X}$, if there exist a function $v(\bm{x})\colon\widehat{\mathcal{X}}\rightarrow \mathbb{R}$ and a bounded function $w(\bm{x})\colon \widehat{\mathcal{X}}\rightarrow \mathbb{R}$ satisfying 
   \begin{equation}
   \label{reach_equa1}
   \begin{cases}
       v(\bm{x})\geq 1-\epsilon_2 & \forall \bm{x}\in \mathcal{X}_0,\\
       v(\bm{x})\leq \mathbb{E}^{\infty}[v(\bm{\phi}^{\bm{x}}_{\pi}(1))]& \forall \bm{x}\in \mathcal{X}\\
       v(\bm{x})\leq \mathbb{E}^{\infty}[w(\bm{\phi}^{\bm{x}}_{\pi}(1))]-w(\bm{x})& \forall \bm{x}\in \mathcal{X}\\
       v(\bm{x})\leq 1 & \forall \bm{x}\in \widehat{\mathcal{X}}\setminus \mathcal{X}
   \end{cases},
   \end{equation}
  then $\mathbb{P}^{\infty}\big(\exists k\in \mathbb{N}. \bm{\phi}^{\bm{x}_0}_{\pi}(k) \in \widehat{\mathcal{X}}\setminus \mathcal{X}\big)\geq v(\bm{x}_0)\geq  1-\epsilon_2$, $\forall \bm{x}_0\in \mathcal{X}_0$.  Consequently, $\mathbb{P}^{\infty}\big(\forall k\in \mathbb{N}. \bm{\phi}^{\bm{x}_0}_{\pi}(k) \in  \mathcal{X}\big)\leq \epsilon_2$, $\forall \bm{x}_0\in \mathcal{X}_0$.
\end{theorem}
\begin{proof}
The conclusion can be assured by following the proof of Corollary 2 in \cite{xue2021reach}.
\end{proof}

By comparing constraints \eqref{super_reach} and \eqref{reach_equa1}, we can conclude when $\mathcal{X}$ is closed and $v(\bm{x})$ is bounded over $\widehat{\mathcal{X}}$ that \eqref{reach_equa1} is weaker than \eqref{super_reach}. This is because if $v(\bm{x})$ satisfies \eqref{super_reach}, $1-v(\bm{x})$ satisfies \eqref{reach_equa1} with $w(\bm{x})=M(1-v(\bm{x}))$ for $\bm{x}\in \widehat{\mathcal{X}}$, where $M\delta \geq  \sup_{\bm{x}\in \widehat{\mathcal{X}}}(1-v(x))$.


According to Theorem \ref{lower_invariance}, an upper bound of the liveness probabilities can be computed via solving \textbf{Op4} (shown in Appendix \ref{sec:appendix_opt}).

\section{Preliminaries on Continuous-time systems}
\label{sec:pre_continuous}
In this section, we first present stochastic continuous-time systems of interest and the safe probabilistic invariance verification problem. After that, we recall the equation in \cite{xue2023reach}, which characterizes the exact probability of reaching target sets eventually while avoiding unsafe states.

\subsection{Problem Statement}
Given a comlete probability space $(\Omega, \mathcal{F}, \mathbb{P})$\cite{oksendal2013stochastic}, a random variable $\bm{X}$ defined on it is an $\mathcal{F}$-measurable function $\bm{X}\colon\Omega \rightarrow \mathbb{R}^n$. A continuous-time stochastic process is a parameterized collection of random variables $\{\bm{X}(t, \bm{w}), t\in T\}$ where the parameter space $T$ can be either the halfline $\mathbb{R}_{\geq 0}$ or an interval $[a, b]$. We consider stochastic systems modeled by time-homogeneous SDEs of the form
\begin{equation}
\label{sde}
    d\bm{X}(t,\bm{w}) = \bm{b}(\bm{X}(t,\bm{w}))dt+\bm{\sigma}(\bm{X}(t,\bm{w}))d\bm{W}(t,\bm{w}),\ t\geq 0,
\end{equation}
where $\bm{X}(\cdot, \cdot)\colon T\times\Omega\rightarrow\mathbb{R}^n$ is an $n$-dimensional continuous-time stochastic process, $\bm{W}(\cdot, \cdot)\colon T\times\Omega\rightarrow\mathbb{R}^m$ is an $m$-dimensional Wiener process (standard Brownian motion), and both mapping $\bm{b}(\cdot)\colon\mathbb{R}^n\rightarrow\mathbb{R}^n$ and $\bm{\sigma}(\cdot)\colon\mathbb{R}^n\rightarrow\mathbb{R}^{n\times m}$ satisfy locally Lipschitz conditions. Given an initial state $\bm{x}_0 \in \mathbb{R}^n$, system \eqref{sde} has a unique (maximal local) strong solution over some time interval $[0, T^{\bm{x}_0}(\bm{w}))$ for $\bm{w} \in \Omega$, where $T^{\bm{x}_0}(\bm{w})$ is a positive real value. We denote it as $\bm{X}^{\bm{x}_0}(\cdot, \bm{w})\colon[0, T^{\bm{x}_0}(\bm{w}))\times\Omega\rightarrow\mathbb{R}^n$, which satisfies the stochastic integral equation
\begin{equation*}
    \begin{split}
\bm{X}^{\bm{x}_0}(t, \bm{w})=\bm{x}_0 & +\int_0^t \bm{b}\left(\bm{X}^{\bm{x}_0}(s, \bm{w})\right) d s +\int_0^t \bm{\sigma}\left(\bm{X}^{\bm{x}_0}(s, \bm{w})\right) d \bm{W}(s, \bm{w}),
    \end{split}
\end{equation*}
for $t \in [0, T^{\bm{x}_0}(\bm{w}))$.

The infinitesimal generator underlying system \eqref{sde}, which characterizes the dynamics of the Itô diffusion at each point in its trajectory, is presented in Definition \ref{def:generator}.

\begin{definition}
    \label{def:generator}
    \cite{oksendal2013stochastic} Let $\boldsymbol{X}^{\boldsymbol{x}}(t, \boldsymbol{w})$ be a time-homogeneous Itô diffusion given by SDE\eqref{sde} with initial state $x \in \mathbb{R}^n$. The infinitesimal generator $\mathcal{A}$ of $\boldsymbol{X}^{\boldsymbol{x}}(t, \boldsymbol{w})$ is   
\begin{equation*}
    \begin{split}
       \mathcal{A} f(\boldsymbol{x})&=\lim _{t \rightarrow 0} \frac{E\left[f\left(\boldsymbol{X}^{\boldsymbol{x}}(t, \boldsymbol{w})\right)\right]-f(\boldsymbol{x})}{t} =\frac{\partial f(\bm{x})}{\partial\bm{x}}\bm{b}(\bm{x})+\frac{1}{2}\text{tr}(\bm{\sigma}(\bm{x})^\top\frac{\partial^2f(\bm{x})}{\partial\bm{x}^2}\bm{\sigma}(\bm{x}))
    \end{split}
\end{equation*}
for any $f \in \mathcal{C}^2\left(\mathbb{R}^n\right)$, where $\mathcal{C}^2\left(\mathbb{R}^n\right)$ denotes the set of twice continuously differentiable functions.
\end{definition}

Next we introduce Dynkin's formula, which provides the expected value of a smooth function of an \textit{Itô} diffusion at a stopping time.

\begin{lemma}[Dynkin's formula, \cite{oksendal2013stochastic}]
\label{theo:dynkin}
    Let $\bm{X}^{\bm{x}}(t, \bm{w})$ be a time-homogeneous \textit{Itô} diffusion given by SDE\eqref{sde} with the initial state $\bm{x} \in \mathbb{R}^n$. Suppose $\tau$ is a stopping time with $E[\tau]<\infty$, and $f \in \mathcal{C}^2\left(\mathbb{R}^n\right)$ with compact support. Then
    \[E\left[f\left(\bm{X}^{\bm{x}}(\tau, \bm{w})\right)\right]=f(\bm{x})+E\left[\int_0^\tau \mathcal{A} f\left(\bm{X}^{\bm{x}}(t, \bm{w})\right) d t\right].\]

\end{lemma}

Similar to Definition \ref{ravoid}, we define the safe probabilistic  invariance verification problem for continuous-time systems.
\begin{definition}
\label{def:ra_continuous}
Given a safe set $\mathcal{X}\subset\mathbb{R}^n$, which is bounded and open, and an initial set $\mathcal{X}_0\subseteq \mathcal{X}$, the safe probabilistic invariance verification for system \eqref{sde} is to compute lower and upper bounds, denoted by $\epsilon_1\in [0,1]$ and $\epsilon_2\in [0,1]$ respectively, for the liveness probabilities that the system, starting from $\mathcal{X}_0$, will remain inside the safe set $\mathcal{X}$ for all time, i.e., to compute $\epsilon_1$ and $\epsilon_2$ such that  
\begin{equation}
\label{eq:ra_verification_continuous}
\epsilon_1\leq \mathbb{P}(
\forall t\in \mathbb{R}_{\geq 0}. \bm{X}^{\bm{x}_0}(t,\bm{w})\in \mathcal{X})\leq \epsilon_2, \forall \bm{x}_0\in \mathcal{X}_0.
\end{equation}
\end{definition}

\subsection{Reachability Probability Characterization in \cite{xue2023reach}}
In this subsection, we will recall an equation, which can characterize the precise probability of system \eqref{sde} entering a specified target set eventually while remaining within a safe set before the first target hitting time.

\begin{proposition}[Theorem 2, \cite{xue2023reach}]
\label{theo:ra_continuous}
Given a bounded and open safe set $\mathcal{X}$ and a target set $\mathcal{X}_r \subseteq \mathcal{X}$, if there exist $v(\bm{x}) \in \mathcal{C}^2(\overline{\mathcal{X}})$ and $u(\bm{x}) \in \mathcal{C}^2(\overline{\mathcal{X}})$ such that for $\bm{x} \in \overline{\mathcal{X}}$,
\begin{equation}
    \label{eq:ra_continuous}
   \begin{cases}
       \widetilde{\mathcal{A}} v(\bm{x}) = 0\\
       v(\bm{x})=1_{\mathcal{X}_r}(\bm{x})+\widetilde{\mathcal{A}}u(\bm{x})
   \end{cases},
\end{equation}
where 
\begin{equation*}
    \widetilde{\mathcal{A}}v(\bm{x})=
    \begin{cases}
    \mathcal{A}v(\bm{x}) & \text{if}\ \bm{x} \in \mathcal{X}\setminus\mathcal{X}_r\\
    0 & \text{if}\ \bm{x} \in \partial\mathcal{X}\cup\mathcal{X}_r
    \end{cases},
    \widetilde{\mathcal{A}}u(\bm{x})=
    \begin{cases}
    \mathcal{A}u(\bm{x}) & \text{if}\ \bm{x} \in \mathcal{X}\setminus\mathcal{X}_r\\
    0 & \text{if}\ \bm{x} \in \partial\mathcal{X}\cup\mathcal{X}_r
    \end{cases}.\\\\
\end{equation*}
Then for $\bm{x}_0\in \mathcal{X}$,
\begin{equation*}\begin{split}
   & \mathbb{P}\big(\exists \tau \in \mathbb{R}_{\geq 0}.  \bm{X}^{\bm{x}_0}(\tau,\bm{w}) \in \mathcal{X}_r\wedge\forall t \in [0, \tau). \bm{X}^{\bm{x}_0}(t,\bm{w}) \in \mathcal{X} \big)\\
   =&\mathbb{P}\big(\exists \tau \in \mathbb{R}_{\geq 0}. \widetilde{\bm{X}}^{\bm{x}_0}(\tau,\bm{w}) \in \mathcal{X}_r\big)=\lim_{\tau\rightarrow \infty}\frac{\mathbb{E}[\int_0^\tau 1_{\mathcal{X}_r}(\widetilde{\bm{X}}^{\bm{x}_0}(t,\bm{w}))dt]}{\tau}=v(\bm{x}_0),
   \end{split}
   \end{equation*} where 
$\widetilde{\bm{X}}^{\bm{x}_0}(t,\bm{w}))$ is a stopped stochastic process 
\begin{equation*}
\widetilde{\bm{X}}^{\bm{x}_0}(t,\bm{w}))=
\begin{cases}
\bm{X}^{\bm{x}_0}(t,\bm{w}) & \text{if}\ t < \tau^{x_0}(\bm{w})\\
\bm{X}^{\bm{x}_0}(\tau^{x_0}(\bm{w}),\bm{w}) & \text{if}\ t \geq \tau^{x_0}(\bm{w})
\end{cases},
\end{equation*}
where $\tau^{x_0}(\bm{w})=\inf \{\,t\mid \bm{X}^{\bm{x}_0}(t,\bm{w})\in \partial\mathcal{X}\vee \bm{X}^{\bm{x}_0}(t,\bm{w})\in \mathcal{X}_r\,\}$ is the first time of exit of $\bm{X}^{\bm{x}_0}(t,\bm{w})$ from the open set $\mathcal{X}\setminus\mathcal{X}_r$.
\end{proposition}

By relaxing \eqref{eq:ra_continuous}, the sufficient conditions for certifying lower and upper bounds of $\mathbb{P}\big(\exists \tau\in \mathbb{R}_{\geq 0}.  \bm{X}^{\bm{x}_0}(\tau,\bm{w}) \in \mathcal{X}_r\wedge \forall t \in [0, \tau). \bm{X}^{\bm{x}_0}(t,\bm{w}) \in \mathcal{X}\big), \forall \bm{x}_0\in \mathcal{X}_0$ are shown in Corollary \ref{prop:lower_ra_continuous} and \ref{prop:upper_ra_continuous}, respectively.

\begin{corollary}
\label{prop:lower_ra_continuous}
Given a bounded and open safe set $\mathcal{X}$, a target set $\mathcal{X}_r$ and an initial set $\mathcal{X}_0$, where $\mathcal{X}_0,\mathcal{X}_r\subseteq \mathcal{X}$, if there exist $v(\bm{x}) \in \mathcal{C}^2(\overline{\mathcal{X}})$ and $u(\bm{x}) \in \mathcal{C}^2(\overline{\mathcal{X}})$ such that,
\begin{equation}
    \label{eq:lower_ra_continuous_1}
   \begin{cases}
       v(\bm{x}) \geq \epsilon_1& \forall \bm{x} \in \mathcal{X}_0\\
       \widetilde{\mathcal{A}} v(\bm{x}) \geq 0&  \forall \bm{x} \in \overline{\mathcal{X}}\\
       v(\bm{x}) \leq 1_{\mathcal{X}_r}(\bm{x})+\widetilde{\mathcal{A}}u(\bm{x}) & \forall \bm{x} \in \overline{\mathcal{X}}
   \end{cases},
\end{equation}
which is equivalent to
\begin{equation}
    \label{eq:lower_ra_continuous_2}
   \begin{cases}
       v(\bm{x}) \geq \epsilon_1 & \forall \bm{x} \in \mathcal{X}_0\\
       \mathcal{A} v(\bm{x}) \geq 0 &\forall \bm{x} \in \mathcal{X}\setminus\mathcal{X}_r  \\
       \mathcal{A} u(\bm{x})-v(\bm{x}) \geq 0 & \forall \bm{x} \in \mathcal{X}\setminus\mathcal{X}_r  \\
       -v(\bm{x})\geq 0 & \forall \bm{x} \in \partial{\mathcal{X}}\\
       1-v(\bm{x})\geq 0 & \forall \bm{x} \in \mathcal{X}_r
   \end{cases},
\end{equation}
then 
\begin{equation*}
    \begin{split}
        &\mathbb{P}\big(\exists \tau \in \mathbb{R}_{\geq 0}.  \bm{X}^{\bm{x}_0}(\tau,\bm{w}) \in \mathcal{X}_r \wedge \forall t \in [0, \tau). \bm{X}^{\bm{x}_0}(t,\bm{w}) \in \mathcal{X} \big)\\
        &=\mathbb{P}\big(\exists \tau\in \mathbb{R}_{\geq 0}. \widetilde{\bm{X}}^{\bm{x}_0}(\tau,\bm{w}) \in \mathcal{X}_r\big)\geq \epsilon_1, \forall \bm{x}_0\in \mathcal{X}_0.
    \end{split}
\end{equation*}
\end{corollary}

\begin{corollary}
\label{prop:upper_ra_continuous}
Given a bounded and open safe set $\mathcal{X}$, a target set $\mathcal{X}_r$ and an initial set $\mathcal{X}_0$, where $\mathcal{X}_0,\mathcal{X}_r\subseteq \mathcal{X}$, if there exist $v(\bm{x}) \in \mathcal{C}^2(\overline{\mathcal{X}})$ and $u(\bm{x}) \in \mathcal{C}^2(\overline{\mathcal{X}})$ such that,
\begin{equation}
    \label{eq:upper_ra_continuous_1}
   \begin{cases}
       v(\bm{x}) \leq \epsilon_2 & \forall \bm{x} \in \mathcal{X}_0\\
       \widetilde{\mathcal{A}} v(\bm{x}) \leq 0 &  \forall \bm{x} \in \overline{\mathcal{X}}\\
       v(\bm{x}) \geq 1_{\mathcal{X}_r}(\bm{x})+\widetilde{\mathcal{A}}u(\bm{x}) &  \forall \bm{x} \in \overline{\mathcal{X}}
   \end{cases},
\end{equation}
which is equivalent to
\begin{equation}
    \label{eq:upper_ra_continuous_2}
   \begin{cases}
       v(\bm{x}) \leq \epsilon_2 & \forall \bm{x} \in \mathcal{X}_0\\
       \mathcal{A} v(\bm{x}) \leq 0 & \forall \bm{x} \in \mathcal{X}\setminus\mathcal{X}_r \\
       \mathcal{A} u(\bm{x}) - v(\bm{x}) \leq 0 &\forall \bm{x} \in \mathcal{X}\setminus\mathcal{X}_r \\
       -v(\bm{x})\leq 0, &\forall \bm{x} \in \partial{\mathcal{X}}\\
       1 - v(\bm{x})\leq 0 & \forall \bm{x} \in \mathcal{X}_r
   \end{cases},
\end{equation}
then \begin{equation*}
    \begin{split}
        &\mathbb{P}\big(\exists \tau \in \mathbb{R}_{\geq 0}.  \bm{X}^{\bm{x}_0}(\tau,\bm{w}) \in \mathcal{X}_r \wedge \forall t \in [0, \tau). \bm{X}^{\bm{x}_0}(t,\bm{w}) \in \mathcal{X} \big)\\
        &=\mathbb{P}\big(\exists \tau\in \mathbb{R}_{\geq 0}. \widetilde{\bm{X}}^{\bm{x}_0}(\tau,\bm{w}) \in \mathcal{X}_r\big)\leq \epsilon_2, \forall \bm{x}_0\in \mathcal{X}_0.
    \end{split}
\end{equation*}
\end{corollary}
\section{Safe Probabilistic Invariance Verification for Continuous-time Systems}
\label{sec:method_continuous}
In this section, we present optimizations to compute the the upper and lower bounds of the liveness probabilities of continuous-time systems. We first present an optimization for computing lower bounds of the liveness probabilities, which is adapted from the classical stochastic barrier certificate based on the Doob's nonnegetive supermartingale inequality. We then present alternative optimizations for lower- and upper-bounding the liveness probabilities, which are adapted from the conditions in Corollary \ref{prop:lower_ra_continuous} and \ref{prop:upper_ra_continuous}.

Similar to Proposition \ref{theo:ra_continuous}, all optimizations are based on a new stochastic process $\{\widehat{\bm{X}}^{\bm{x}_0}(t, \bm{w}),t\in \mathbb{R}_{\geq 0}\}$ for $\bm{x}_0 \in \overline{\mathcal{X}}$, which is a stopped process corresponding to $\{\bm{X}^{\bm{x}_0}(t, \bm{w}),t\in [0,T^{\bm{x}_0}(\bm{w}))\}$ and the set $\mathcal{X}$, i.e.,
\begin{equation}
\label{eq:switch_continuous}
\widehat{\bm{X}}^{\bm{x}_0}(t,\bm{w}))=
\begin{cases}
\bm{X}^{\bm{x}_0}(t,\bm{w}) &\text{if}\ t < \tau^{x_0}(\bm{w})\\
\bm{X}^{\bm{x}_0}(\tau^{x_0}(\bm{w}),\bm{w}) &\text{if}\ t \geq \tau^{x_0}(\bm{w})
\end{cases},
\end{equation}
where $\tau^{x_0}(\bm{w})=\inf \{\,t\mid \bm{X}^{\bm{x}_0}(t,\bm{w}))\in \partial\mathcal{X}\,\}$ is the first time of exit of $\bm{X}^{\bm{x}_0}(t,\bm{w})$ from the open set $\mathcal{X}$. Clearly, the set $\overline{\mathcal{X}}$ is an invariant set for $\widehat{\bm{X}}^{\bm{x}_0}(t, \bm{w})$ with $\bm{x}_0 \in \overline{\mathcal{X}}$. Moreover, the infinitesimal generator corresponding to $\widehat{\bm{X}}^{\bm{x}_0}(t,\bm{w})$ (denoted $\widehat{\mathcal{A}}$) is 
identical to the one corresponding to $\bm{X}^{\bm{x}_0}(t,\bm{w})$ on the set $\mathcal{X}$, and is equal to zero on the set $\partial\mathcal{X}$, i.e.,
\begin{equation*}
\widehat{\mathcal{A}}f(\bm{x})=
\begin{cases}
\mathcal{A}f(\bm{x}) & \text{if}\ \bm{x} \in \mathcal{X}\\
0 & \text{if}\ \bm{x} \in \partial\mathcal{X}
\end{cases}.\\
\end{equation*}

Next, we clarify the auxiliary role of $\widehat{\bm{X}}^{\bm{x}_0}(t, \bm{w})$ in the calculation of the liveness probability. Given an initial state $\bm{x}_0 \in \mathcal{X}_0$, we have $\bm{X}^{\bm{x}_0}(t, \bm{w}) = \widehat{\bm{X}}^{\bm{x}_0}(t, \bm{w}),\forall t \in [0, \tau]$. Thus, 
\begin{equation*}
\begin{split}
    &\mathbb{P}\big(\exists \tau \in \mathbb{R}_{\geq 0}.\bm{X}^{\bm{x}_0}(\tau, \bm{w}) \in \partial\mathcal{X} \wedge \forall t \in [0, \tau). \bm{X}^{\bm{x}_0}(t, \bm{w}) \in \mathcal{X}\big)\\
    =&\mathbb{P}\big(\exists \tau \in \mathbb{R}_{\geq 0}.\widehat{\bm{X}}^{\bm{x}_0}(\tau, \bm{w}) \in \partial\mathcal{X} \wedge \forall t \in [0, \tau). \widehat{\bm{X}}^{\bm{x}_0}(t, \bm{w}) \in \mathcal{X}\big)
\end{split}
\end{equation*}
and 
\begin{equation*}
    \mathbb{P}\big(\forall t \in \mathbb{R}_{\geq 0}.\bm{X}^{\bm{x}_0}(t, \bm{w}) \in \mathcal{X}\big) = \mathbb{P}\big(\forall t \in \mathbb{R}_{\geq 0}.\widehat{\bm{X}}^{\bm{x}_0}(t, \bm{w}) \in \mathcal{X}\big). 
\end{equation*}

Similarly, we can obtain the lower and upper bounds of the liveness probability of staying inside the safe set $\mathcal{X}$ by computing the upper and lower bounds of the exit probability $\mathbb{P}\big(\exists \tau \in \mathbb{R}_{\geq 0}.\widehat{\bm{X}}^{\bm{x}_0}(\tau, \bm{w}) \in \partial\mathcal{X} \wedge \forall t \in [0, \tau). \widehat{\bm{X}}^{\bm{x}_0}(t, \bm{w}) \in \mathcal{X}\big)$.


\subsection{Doob's Nonnegative Supermartingale Inequality Based Invariance Verification}
This subsection formulates a sufficient condition for lower-bounding the liveness probabilities in Definition \ref{def:ra_continuous}. This condition is the straightforward extension of the one in Theorem 15 in \cite{prajna2007framework} and the one in Proposition III.8 in \cite{wang2021safety}, which is built upon the well established Doob's nonnegative supermartingale inequality \cite{ville1939etude}. This condition is formulated here primarily for convenient comparisons throughout the remainder.  


\begin{proposition}[Theorem 15, \cite{prajna2007framework}]
\label{prop:origin_2007}
Given a bounded and open safe set $\mathcal{X}$ and an initial set $\mathcal{X}_0 \subseteq \mathcal{X}$, if there exist $v(\bm{x}) \in \mathcal{C}^2({\mathcal{X}})$, such that,
\begin{equation}
    \label{eq:lower_2007}
   \begin{cases}
       v(\bm{x}) \leq 1 - \epsilon_1 & \forall \bm{x} \in \mathcal{X}_0\\
       \mathcal{A} v(\bm{x}) \leq 0 &\forall \bm{x} \in \overline{\mathcal{X}} \\
       v(\bm{x}) \geq 0 &\forall \bm{x} \in \overline{\mathcal{X}} \\
       v(\bm{x})\geq 1 & \forall \bm{x} \in \mathcal{X}_{unsafe} (= \partial \mathcal{X})
   \end{cases},
\end{equation}
then $\mathbb{P}\big(\exists t\in \mathbb{R}_{\geq 0}. \bm{X}^{\bm{x}_0}(t,\bm{w}) \in \partial \mathcal{X} \big)
   \leq 1-\epsilon_1$ and thus $\mathbb{P}\big(\forall t\in \mathbb{R}_{\geq 0}. \bm{X}^{\bm{x}_0}(t,\bm{w}) \in \mathcal{X} \big) \geq 1 - v(\bm{x}_0) \geq \epsilon_1$, $\forall \bm{x}_0\in \mathcal{X}_0$.
\end{proposition}

According to Proposition \ref{prop:origin_2007}, a lower bound of the liveness probabilities can be computed via solving \textbf{Op5} (shown in Appendix \ref{sec:appendix_opt}).

To the best of our knowledge, there are no barrier-like conditions based on the Doob's nonnegative supermartingale inequality, similar to the one stated in Theorem 16 of \cite{anand2022k} for discrete-time systems, that have been developed in previous studies to examine upper bounds of the reachability probabilities. Moreover, it is not possible to adapt the constraint $ \mathbb{E}^{\infty}[v(\bm{\phi}_{\pi}^{\bm{x}}(1))]-v(\bm{x}) \leq -\delta, \forall \bm{x} \in \mathcal{X}$ from equation \eqref{super_reach} to $ \mathcal{A}v(\bm{x}) \leq -\delta, \forall \bm{x} \in \mathcal{X}$ in order to compute upper bounds for the liveness probabilities. If we consider that $ \mathcal{A}v(\bm{x}) \leq -\delta, \forall \bm{x} \in \mathcal{X}$ is true, then $\widehat{\mathcal{A}}v(\bm{x}) \leq -\delta, \forall \bm{x} \in \mathcal{X}$ will also hold, which implies $\widehat{\mathcal{A}}v(\bm{x}) \leq -\delta, \forall \bm{x} \in \overline{\mathcal{X}}$. However, this contradicts the condition $\widehat{\mathcal{A}}v(\bm{x})=0, \forall \bm{x} \in \partial \mathcal{X}$.

\subsection{Equations Relaxation Based Invariance Verification}
In this subsection, we present the second set of optimizations for addressing the safe probabilistic invariance verification problem in Definition \ref{def:ra_continuous}. We begin by introducing an equation that characterizes the liveness probability of satying the safe set $\mathcal{X}$ for all the time. This equation is adapted from \eqref{eq:ra_continuous}, where the unsafe set $\partial\mathcal{X}$ is regarded as the target set $\mathcal{X}_r$ in \eqref{eq:ra_continuous}. We then propose two sufficient conditions for certifying lower and upper bounds of the liveness probabilities by relaxing the derived equation.

\begin{lemma}
\label{theo:continuous}
Given a bounded safe set $\mathcal{X}$, if there exist $v(\bm{x}) \in \mathcal{C}^2(\overline{\mathcal{X}})$ and $u(\bm{x}) \in \mathcal{C}^2(\overline{\mathcal{X}})$ such that for $\bm{x} \in \overline{\mathcal{X}}$,
\begin{equation}
    \label{eq:continuous}
   \begin{cases}
       \widehat{\mathcal{A}} v(\bm{x}) = 0\\
       v(\bm{x})=1_{\partial\mathcal{X}}(\bm{x})+\widehat{\mathcal{A}}u(\bm{x})
   \end{cases},
\end{equation}
then
\begin{equation*}\begin{split}
   v(\bm{x}_0) &= \mathbb{P}\big(\exists t\in \mathbb{R}_{\geq 0}. \bm{X}^{\bm{x}_0}(t,\bm{w}) \in \partial\mathcal{X} \big)\\
   &= \mathbb{P}\big(\exists t\in \mathbb{R}_{\geq 0}. \widehat{\bm{X}}^{\bm{x}_0}(t,\bm{w}) \in \partial\mathcal{X} \big)\\
   &=\lim_{t\rightarrow \infty}\frac{\mathbb{E}[\int_0^t 1_{\partial\mathcal{X}}(\widehat{\bm{X}}^{\bm{x}_0}(t,\bm{w}))d\tau]}{t}, \forall \bm{x}_0\in \mathcal{X}.
   \end{split}
   \end{equation*}
    Thereby, $\mathbb{P}\big(\forall t\in \mathbb{R}_{\geq 0}. \bm{X}^{\bm{x}_0}(t,\bm{w}) \in \mathcal{X}\big) = 1 - v(\bm{x}_0)$, $\forall \bm{x}_0\in \mathcal{X}$.
\end{lemma}
\begin{proof}
The conclusion can be assured by following the proof of Theorem 2 in \cite{xue2023reach}.
\end{proof}

Similar to the Section \ref{SIV}, two sufficient conditions can be obtained for certifying lower and upper bounds of the liveness probabilities via directly relaxing equation \eqref{eq:continuous}, respectively.

\begin{theorem}
\label{prop:lower_continuous}
Given a bounded and open safe set $\mathcal{X}$ and an initial set $\mathcal{X}_0$, where $\mathcal{X}_0 \subseteq \mathcal{X}$, if there exist $v(\bm{x}) \in \mathcal{C}^2(\overline{\mathcal{X}})$ and $u(\bm{x}) \in \mathcal{C}^2(\overline{\mathcal{X}})$ such that
\begin{equation}
    \label{eq:lower_continuous}
   \begin{cases}
       v(\bm{x}) \leq 1 - \epsilon_1 &\forall \bm{x} \in \mathcal{X}_0 \\
       \mathcal{A} v(\bm{x}) \leq 0 &\forall \bm{x} \in \mathcal{X}\\
       v(\bm{x}) \geq \mathcal{A} u(\bm{x}) & \forall \bm{x} \in \mathcal{X}\\
       v(\bm{x})\geq 1 & \forall \bm{x} \in \partial{\mathcal{X}}
   \end{cases},
\end{equation}
then $\mathbb{P}\big(\exists t\in \mathbb{R}_{\geq 0}. \bm{X}^{\bm{x}_0}(t,\bm{w}) \in \partial\mathcal{X} \big)\leq v(\bm{x}_0) \leq 
 1- \epsilon_1$, $\forall \bm{x}_0\in \mathcal{X}_0$. Hence, $\mathbb{P}\big(\forall t\in \mathbb{R}_{\geq 0}. \bm{X}^{\bm{x}_0}(t,\bm{w}) \in \mathcal{X} \big)\geq \epsilon_1$, $\forall \bm{x}_0\in \mathcal{X}_0$.
\end{theorem}
\begin{proof}
The conclusion can be assured by following the proof of Corollary 2 in \cite{xue2023reach}.
\end{proof}

Similarly, we can show that constraints \eqref{eq:lower_2007} and \eqref{eq:lower_continuous} are equivalent. 

\begin{proposition}
    \label{coro:continuous}
    Constraints \eqref{eq:lower_2007} and \eqref{eq:lower_continuous} are equivalent.
\end{proposition}
\begin{proof}
    The proof is similar to Proposition \ref{coro:discrete}, which is shown in Appendix \ref{sec:appendix_proof}.
\end{proof}

According to Theorem \ref{prop:lower_continuous}, a lower bound of the liveness probabilities can be computed via solving \textbf{Op6} (shown in Appendix \ref{sec:appendix_opt}).

The sufficient condition for certifying upper bounds of the liveness probabilities is presented in Theorem \ref{prop:upper_continuous}.
\begin{theorem}
\label{prop:upper_continuous}
Given a bounded and open safe set $\mathcal{X}$ and an initial set $\mathcal{X}_0$, where $\mathcal{X}_0 \subseteq \mathcal{X}$, if there exist $v(\bm{x}) \in \mathcal{C}^2(\overline{\mathcal{X}})$ and $u(\bm{x}) \in \mathcal{C}^2(\overline{\mathcal{X}})$ such that,
\begin{equation}
    \label{eq:upper_continuous}
   \begin{cases}
       v(\bm{x}) \geq 1 - \epsilon_2 & \forall \bm{x} \in \mathcal{X}_0\\
       \mathcal{A} v(\bm{x}) \geq 0 & \forall \bm{x} \in \mathcal{X} \\
       v(\bm{x}) \leq \mathcal{A} u(\bm{x}) & \forall \bm{x} \in \mathcal{X} \\
       v(\bm{x})\leq 1 & \forall \bm{x} \in \partial{\mathcal{X}}
   \end{cases},
\end{equation}
then $\mathbb{P}\big(\exists t\in \mathbb{R}_{\geq 0}. \bm{X}^{\bm{x}_0}(t,\bm{w}) \in \partial\mathcal{X} \big)\geq v(\bm{x}_0) \geq 
 1- \epsilon_2$, $\forall \bm{x}_0\in \mathcal{X}_0$. Hence, $\mathbb{P}\big(\forall t\in \mathbb{R}_{\geq 0}. \bm{X}^{\bm{x}_0}(t,\bm{w}) \in \mathcal{X} \big)\leq \epsilon_2$, $\forall \bm{x}_0\in \mathcal{X}_0$.
\end{theorem}
\begin{proof}
The conclusion can be assured by following the proof of Corollary 1 in \cite{xue2023reach}.
\end{proof}

According to Theorem \ref{prop:upper_continuous}, an upper bound of the liveness probabilities can be computed via solving \textbf{Op7} (shown in Appendix \ref{sec:appendix_opt}).
\section{Examples}
\label{sec:example}
In this section we demonstrate our theoretical developments on five examples. Since directly solving the problems \textbf{Op0}-\textbf{Op8} is challenging, we relax them into semi-definite programs (\textbf{SDP0}-\textbf{SDP8}, see Appendix \ref{sec:appendix_sdp}) using the sum of squares decomposition for multivariate polynomials, which allows for efficient solving. \oomit{Nonetheless, this relaxation may yield unexpected outcomes.} We solve the resulting semi-definite programs using the tool Mosek 10.1.21 \cite{aps2019mosek}. Furthermore, to ensure numerical stability during the solution of the semi-definite programs, we impose a constraint on the coefficients of the unknown polynomials ($v(\bm{x}),w(\bm{x}),u(\bm{x}),p(\bm{x}), s_i(\bm{x}),i=0,\ldots,4$) in \textbf{SDP0}-\textbf{SDP8} and restrict them to the interval $[-100,100]$.

In addition, in the experiments we employ the Monte Carlo method to approximate the upper and lower bounds of the liveness probabilities. Specifically, we take a large number of samples (e.g., $10^4$) from the initial set and simulate the liveness probabilities of these samples over an extended time. The maximum probability among these samples is considered as the upper bound, while the minimum probability is considered as the lower bound. By using these bounds as ground truth values, we can assess the accuracy of the bounds computed via solving \textbf{SDP0}-\textbf{SDP8}. 




\begin{table*}[h]
    \centering
    \begin{tabular}{|c|c|c|c|c|c|c|c|}\hline 
           d &2&6&10&14&18&22&26  \\\hline
        $\epsilon_1$ by Op3 &0.3574& 0.6678& 0.6917& 0.7368& 0.7575& 0.7622& 0.7647\\\hline 
$\epsilon_2$ by Op4 &1.0000& 0.9505& 0.9488& 0.9242& 0.8991& 0.8927& 0.8771\\\hline
        $\epsilon_1$ by Op1 &0.3574& 0.6678& 0.6917& 0.7368& 0.7575& 0.7622& 0.7647 \\\hline 
        $\epsilon_1$ by Op0 &0.3574& 0.5895& 0.5937& 0.6740& 0.6867& 0.7007& 0.7284 \\\hline 
    \end{tabular}
    \caption{Computed lower and upper bounds of the liveness probability in Example \ref{ex:one_dim}\\
    ($d$ denotes the degree of unknown polynomials involved in the resulting SDPs)}
    \label{tab:my_label}
\end{table*}

\subsection{Discrete-time Systems}
\begin{example}
\label{ex:one_dim}
Consider the one-dimensional discrete-time system:
\begin{equation*}
x(l+1)=(-0.5+d(l))x(l),
\end{equation*}
where $d(\cdot)\colon\mathbb{N}\rightarrow \mathcal{D}=[-1,1]$, $\mathcal{X}=\{\,x\mid h(x)\leq 0\,\}$ with $h(x)=x^2-1$, and $\mathcal{X}_0=\{\,x\mid (x+0.8)^2=0\,\}$. Besides, we assume that the probability distribution on $\mathcal{D}$ is the uniform distribution. The lower and upper bounds of the liveness probability obtained by Monte Carlo methods are $\epsilon_1 = \epsilon_2 = 0.8312$.
\end{example}

The set $\widehat{\mathcal{X}}=\{\,x\mid \widehat{h}(x)\leq 0\,\}$ with $\widehat{h}(x)=x^2-2.25$ is used in  solving \textbf{SDP1}-\textbf{SDP4}.  The computed lower and upper bounds are summarized in Table \ref{tab:my_label}\oomit{ and Fig. \ref{fig_1}}. It is concluded that tighter lower and upper bounds of the liveness probability can be obtained when polynomials of higher degree are used for performing computations. 
Given that the liveness probability is strictly greater than zero, as stated in Remark \ref{remark_bounded}, the use of \textbf{SDP2} for computing upper bounds on the liveness probability is precluded. This observation aligns with our experimental findings. Instead, upper bounds for the liveness probability can be obtained by solving \textbf{SDP4}, which closely approximate the results obtained from Monte-Carlo methods as the degree of polynomials increases. Concurrently, we note a noteworthy consistency in the identical lower bounds computed from \textbf{SDP1} and \textbf{SDP3}, with those derived from \textbf{SDP0} being the most conservative. The identical lower bounds computed from \textbf{SDP1} and \textbf{SDP3} also justify the validity of the conclusion in Proposition \ref{coro:discrete}.


\begin{example}
\label{lvm}
Consider the discrete-time Lotka-Volterra model:
\begin{equation}
\label{volterra}
    \begin{cases}
x(l+1)=rx(l)-ay(l)x(l),\\
y(l+1)=sy(l)+acy(l)x(l),
\end{cases}
\end{equation}
where $r=0.5$, $a=1$, $s=-0.5+d(l)$ with $d(\cdot)\colon\mathbb{N}\rightarrow \mathcal{D}=[-1,1]$ and $c=1$, $\mathcal{X}=\{\,(x,y)^{\top}\mid h(\bm{x})\leq 0\,\}$ with $h(\bm{x})=x^2+y^2-1$, and $\mathcal{X}_0=\{\,(x,y)^{\top}\mid g(\bm{x})\leq 0\,\}$ with $g(\bm{x})=(x+0.6)^2+(y+0.6)^2-0.01$. Besides, we assume that the probability distribution imposed on $\mathcal{D}$ is the uniform distribution.
\end{example}

The set $\widehat{\mathcal{X}}=\{\,(x,y)^{\top}\mid \widehat{h}(\bm{x})\leq 0\,\}$ with $\widehat{h}(\bm{x})=x^2+y^2-4$ is used in solving \textbf{SDP1}-\textbf{SDP4}. The lower and upper bounds of the liveness probabilities obtained by Monte Carlo methods are $\epsilon_1 = 0.3651$ and $\epsilon_2 = 0.7974$, respectively. Five trajectories starting from $(-0.6,-0.6)^{\top}$ are visualized in Fig. \ref{fig_2}. The computed lower and upper bounds are presented in Table \ref{tab2}. Analogous to Example \ref{ex:one_dim}, \textbf{SDP2} cannot be used to compute upper bounds of the liveness probabilities. Instead, upper bounds can be obtained by solving \textbf{SDP4}, which closely approximates the results obtained from Monte Carlo methods as the degree of polynomials increases. In addition, the lower bounds obtained by solving \textbf{SDP3} and \textbf{SDP1} are almost equal and tighter than those obtained by solving \textbf{SDP0}. Fig. \ref{fig:ex_lvm_v} depicts a comparison of the computed $v(\bm{x})$ when employing polynomials of degree 18. The $v(\bm{x})$ obtained from \textbf{SDP3} is greater than or equal to zero for every $\bm{x} \in \mathcal{X}$ and closely resembles the $v(\bm{x})$ obtained from \textbf{SDP1}. Unlike \textbf{Op1} and \textbf{Op3}, which impose a non-negative requirement on $v(\bm{x})$ for each $\bm{x}\in \widehat{\mathcal{X}}$, \textbf{Op4} does not have this requirement and allows $v(\bm{x})$ to be negative for some $\bm{x}\in \widehat{\mathcal{X}}$. This is illustrated in Fig. \ref{fig:ex_lvm_v}.   


\begin{figure}
    \centering
    \includegraphics[width=0.9\textwidth]{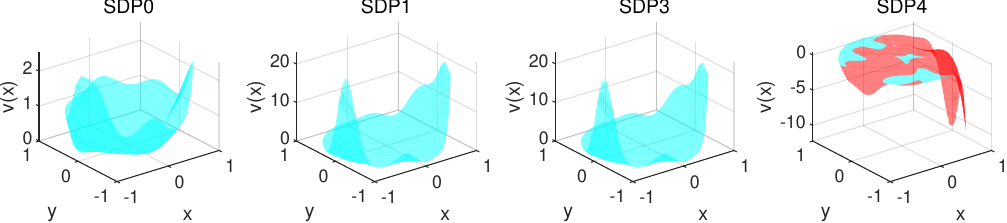}
    \caption{Illustration of computed $v(\bm{x})$ on $\mathcal{X}$ when $d=18$ in Example \ref{lvm}. Red region denotes $v(\bm{x})<0$, while blue region denotes $v(\bm{x})\geq 0$.}
    \label{fig:ex_lvm_v}
\end{figure}
   

\begin{table}[h]
    \centering
    \begin{tabular}{|c|c|c|c|c|c|c|}\hline 
           d &$8$&$10$&12 &14 & 16 & 18  \\\hline
        $\epsilon_1$ by Op3 &0.1487 &0.1865 &0.2039 &0.2147 &0.2203 &0.2261\\\hline 
        $\epsilon_2$ by Op4 &0.9064 &0.8670 &0.8501 &0.8320 &0.8272 &0.8171\\\hline
        $\epsilon_1$ by Op1 &0.1487 &0.1865 &0.2039 &0.2147 &0.2203 &0.2260\\\hline 
        $\epsilon_1$ by Op0 &0.0000 &0.0000 &0.0000 &0.0000 &0.1756 &0.1976\\\hline 
    \end{tabular}
    \caption{Computed lower and upper bounds of the liveness probabilities in Example \ref{lvm} ($d$ denote the degree of unknown polynomials involved in the resulting SDPs)}
    \label{tab2}
\end{table}

    \begin{figure}
    \centering
    \begin{subfigure}[b]{0.24\textwidth}
   \includegraphics[width=\textwidth, trim=55 0 70 23, clip]{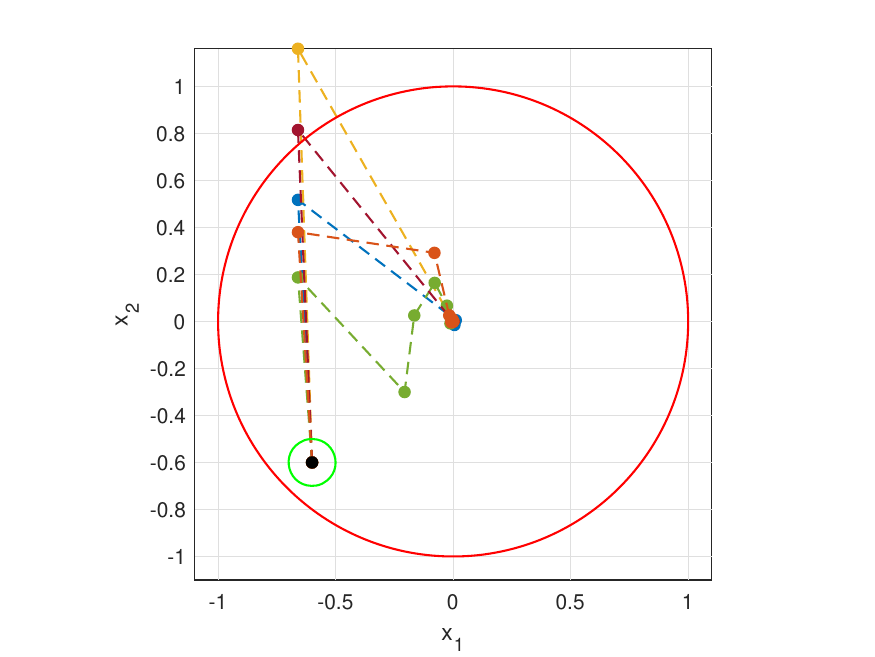}
    \caption{Trajectories for system \eqref{volterra} starting from $(-0.6,-0.6)^{\top}$.}
      \label{fig_2}
    \end{subfigure}
    \hfill
    \begin{subfigure}[b]{0.24\textwidth}
    \centering
    \includegraphics[width=\textwidth, trim=55 0 70 23, clip]{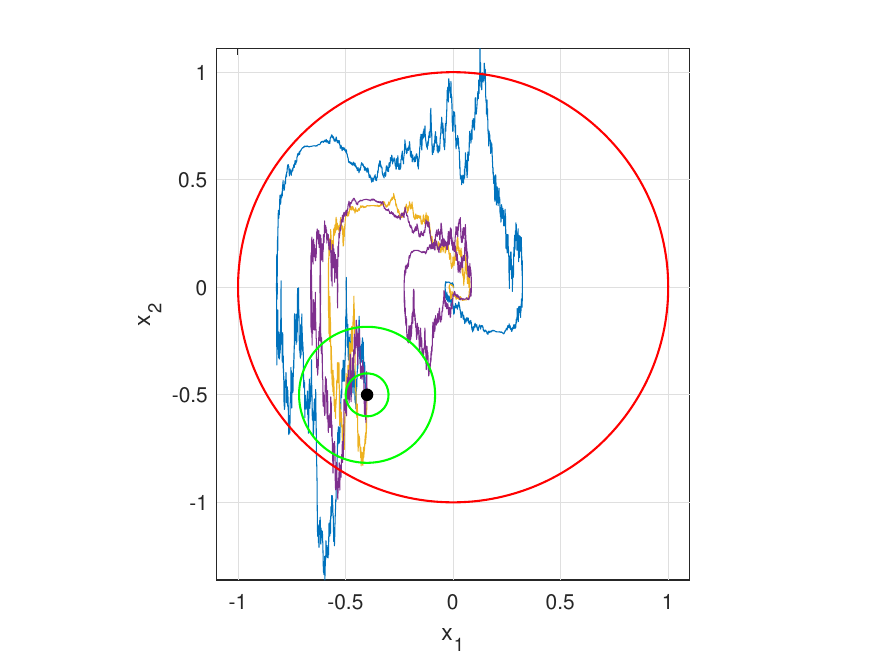}
    \caption{Trajectories for system \eqref{eq:drift} starting from $(-0.4,0.5)^{\top}$.}
      \label{fig:drift}
    \end{subfigure}
    \hfill
    \begin{subfigure}[b]{0.49\textwidth}
            \centering
    \includegraphics[width=\textwidth, trim=35 0 37 20, clip]{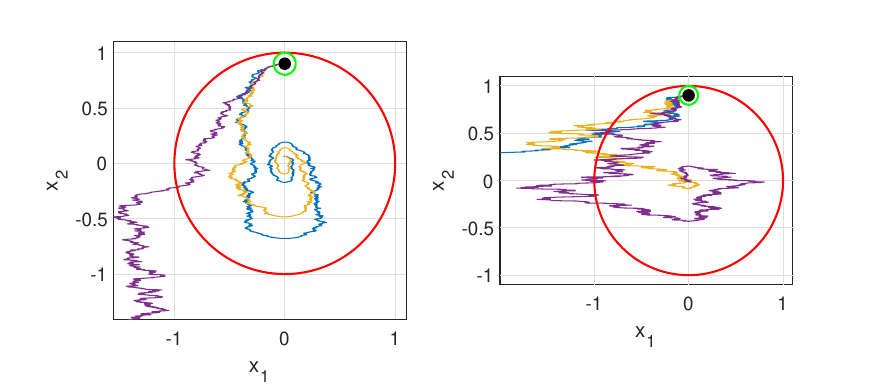}
    \caption{Trajectories for system \eqref{eq:tan} starting from $(0,0.9)^{\top}$ with $\sigma=0.5$ in the left and $\sigma=2$ in the right.}
    \label{fig:tan}
    \end{subfigure}
    \caption{Illustrations of trajectories for systems (red curve - $\partial \mathcal{X}$; green curve - $\partial \mathcal{X}_0$).}
  \end{figure}

\subsection{Continuous-time Systems}
\begin{example}[Nonlinear drift]
    \label{ex:drift}
    Consider the stochastic differential equation that is adapted from \cite{prajna2004stochastic},
    \begin{equation}
        \label{eq:drift}
        \begin{cases}
            dX_1(t,w) = X_2(t,w)dt,\\
            dX_2(t,w) = -(X_1(t,w) + X_2(t,w) + 0.5X_1^3(t,w))dt\\
            \qquad\qquad\quad+(X_1(t,w) + X_2(t,w))dW(t,w)
        \end{cases},
    \end{equation}
  where the safe set is $\mathcal{X} = \{\,(x_1,x_2)^\top \mid h(\bm{x})<0\,\}$ with $h(\bm{x})= x_1^2+x_2^2-1$.  
  \end{example}
  In this example, we consider two different initial sets: $\mathcal{X}_0 = \{\,(x_1,x_2)^\top \mid g_1(\bm{x})<0\,\}$ with $g_1(\bm{x})= 100(x_1+0.4)^2+100(x_2+0.5)^2-1$ and $\mathcal{X}_0^\prime = \{\,(x_1,x_2)^\top \mid g_2(\bm{x})<0\}$ with $g_2(\bm{x})= 10(x_1+0.4)^2+10(x_2+0.5)^2-1$.  The lower and upper bounds of the liveness probabilities obtained through Monte Carlo methods are $\epsilon_1 = 0.5338$ and $\epsilon_2 = 0.7101$, respectively, when adopting $\mathcal{X}_0$. Conversely, when adopting $\mathcal{X}_0^\prime$, the lower and upper bounds are $\epsilon_1 = 0.2114$ and $\epsilon_2 = 0.8976$, respectively. Three trajectories starting from $(-0.4,0.5)^{\top}$ are visualized in Fig. \ref{fig:drift}.

 The lower and upper bounds computed by solving \textbf{SDP5} through \textbf{SDP7} are summarized in Table \ref{tab:drift1} and Table \ref{tab:drift2}. The results demonstrate that the lower bounds of the liveness probabilities obtained from \textbf{SDP5} and \textbf{SDP6} are equal, regardless of the initial sets. This aligns with the conclusion stated in Proposition \ref{coro:continuous}. Meanwhile, the upper bounds can be computed by solving \textbf{SDP7}, which are close to the ones from Monte Carlo methods. Furthermore, due to the inclusion relationship of $\mathcal{X}_0 \subset \mathcal{X}_0^\prime$,  the lower bounds using $\mathcal{X}_0^\prime$ are less than or equal to the lower bounds obtained with $\mathcal{X}_0$, while the upper bounds using $\mathcal{X}_0^\prime$ are greater than or equal to the upper bounds computed using $\mathcal{X}_0$. These findings are consistent with our experimental results, as presented in Table \ref{tab:drift1} and Table \ref{tab:drift2}.


  \begin{table}[h]
    \centering
    \begin{tabular}{|c|c|c|c|c|c|c|c|}\hline 
           d &4&6&8&10&12&14&16  \\\hline
        $\epsilon_1$ by Op6 &0.3957 & 0.4217 & 0.4590 &0.4660 &0.4675 &0.4682 &0.4686\\\hline 
$\epsilon_2$ by Op7 &0.7313 &0.7279 &0.7233 &0.7224 & 0.7216 & 0.7213 &0.7208\\\hline
        $\epsilon_1$ by Op5 & 0.3957 & 0.4217 & 0.4590 &0.4660 &0.4675 &0.4682 &0.4686 \\\hline 
    \end{tabular}
    \caption{Computed lower and upper bounds of the liveness probabilities in Example \ref{ex:drift} using $\mathcal{X}_0$ ($d$ denotes the degree of unknown polynomials involved in the SDPs)}
    \label{tab:drift1}
\end{table}

  \begin{table}[h]
    \centering
    \begin{tabular}{|c|c|c|c|c|c|c|c|}\hline 
           d &4&6&8&10&12&14&16  \\\hline
        $\epsilon_1$ by Op6 &0.0596 & 0.0721 & 0.0873 &0.0913 &0.0919 &0.0923 &0.0924\\\hline 
$\epsilon_2$ by Op7 &0.9046 &0.9032 &0.9011 &0.9007 & 0.9004 & 0.9002 &0.9000\\\hline
        $\epsilon_1$ by Op5 &0.0596 & 0.0721 & 0.0873 &0.0913 &0.0919 &0.0923 &0.0924 \\\hline 
    \end{tabular}
    \caption{Computed lower and upper bounds of the liveness probabilities in Example \ref{ex:drift} using $\mathcal{X}_0^\prime$ ($d$ denotes the degree of unknown polynomials involved in the SDPs)}
    \label{tab:drift2}
\end{table}

\begin{example}
    \label{ex:tan}
    Consider the following nonlinear stochastic system from \cite{tan2008stability},
    \begin{equation}
        \label{eq:tan}
        \begin{cases}
            dX_1(t,w) = &\big(-0.42X_1(t,w)-1.05X_2(t,w)- 2.3X_1^2(t,w)  \\
            &- 0.56X_1(t,w)X_2(t,w)- X_1^3(t,w)\big)dt + \sigma X_1(t,w) dW(t,w)\\
            dX_2(t,w) = & (1.98X_1(t,w) + X_1(t,w)X_2(t,w))dt
        \end{cases},
    \end{equation}
    where the safe set is $\mathcal{X} = \{\,(x_1,x_2)^\top \mid h(\bm{x})<0\,\}$ with $h(\bm{x})= x_1^2+x_2^2-1$ and the initial set is $\mathcal{X}_0 = \{\,(x_1,x_2)^\top \mid g(\bm{x})<0\,\}$ with $g(\bm{x})= x_1^2+(x_2-0.9)^2-0.01$.
    \end{example}
In this example, we consider different values of $\sigma$: $\sigma = 0.5$ and $\sigma = 2$. When $\sigma = 0.5$, the lower and upper bounds of the liveness probabilities obtained through Monte Carlo methods are $\epsilon_1=0.0000$ and $\epsilon_2=0.4398$, respectively; when $\sigma = 2$, $\epsilon_1=0.0000$ and $\epsilon_2=0.3580$. For these two scenarios, the lower and upper bounds of the liveness probabilities computed by solving \textbf{SDP5} to \textbf{SDP7} are summarized in Table \ref{tab:tan1} and Table \ref{tab:tan2}, respectively. Additionally, Fig. \ref{fig:tan} visualizes three trajectories starting from $(0,0.9)^{\top}$.

    Regardless of the values of $\sigma$, the lower bounds obtained from \textbf{SDP5} and \textbf{SDP6} are zero. Meanwhile, the upper bounds of the liveness probabilities, computed by solving \textbf{SDP7}, vary for these two values of $\sigma$. It is observed that as the degree $d$ of polynomials used  for computations increases, the upper bounds decrease and approach the one obtained from Monte Carlo methods.
    

  \begin{table}[h]
    \centering
    \begin{tabular}{|c|c|c|c|c|c|c|c|}\hline 
           d &18&20&22&24&26&28&30  \\\hline
        $\epsilon_1$ by Op6&0.0000 &0.0000 & 0.0000 & 0.0000 &0.0000 &0.0000 &0.0000\\\hline 
$\epsilon_2$ by Op7 &0.5829 &0.5768 &0.5709 &0.5654 &0.5604 & 0.5574 &0.5560\\\hline
        $\epsilon_1$ by Op5 &0.0000 &0.0000 & 0.0000 & 0.0000 &0.0000 &0.0000 &0.0000\\\hline 
    \end{tabular}
    \caption{Computed lower and upper bounds of the liveness probabilities in Example \ref{ex:tan} with $\sigma = 0.5$ ($d$ denotes the degree of unknown polynomials involved in the SDPs)}
    \label{tab:tan1}
\end{table}

  \begin{table}[h]
    \centering
    \begin{tabular}{|c|c|c|c|c|c|c|c|}\hline 
           d &18&20&22&24&26&28&30  \\\hline
        $\epsilon_1$ by Op6 &0.0000 &0.0000 & 0.0000 & 0.0000 &0.0000 &0.0000 &0.0000\\\hline 
$\epsilon_2$ by Op 7&0.3731 &0.3723 &0.3718 &0.3715 &0.3711 &0.3709 & 0.3707\\\hline
        $\epsilon_1$ by Op5 &0.0000 &0.0000 & 0.0000 & 0.0000 &0.0000 &0.0000 &0.0000\\\hline 
    \end{tabular}
    \caption{Computed lower and upper bounds of the liveness probabilities in Example \ref{ex:tan} with $\sigma = 2$ ($d$ denotes the degree of unknown polynomials involved in the SDPs)}
    \label{tab:tan2}
\end{table}

\begin{example}
    \label{ex:vanderpol}
    Consider a 3D VanderPol oscillator adapted from \cite{korda2014controller},
    \begin{equation*}
        \begin{cases}
            dX_1(t,w) = & -2X_2(t,w)dt + X_1(t,w) dW(t,w)\\
            dX_2(t,w) = & \big(0.8X_1(t,w)-2.1X_2(t,w)+X_3(t,w)+10X_1^2(t,w)X_2(t,w)\big)dt\\
            &+ X_2(t,w) dW(t,w)\\
            dX_3(t,w) = & \big(-X_3(t,w)+X_3^3(t,w)\big)dt + X_3(t,w) dW(t,w)\\
        \end{cases},
    \end{equation*}
    where the safe set is $\mathcal{X} = \{\,(x_1,x_2,x_3)^\top \mid h(\bm{x})<0\,\}$ with $h(\bm{x})= x_1^2+x_2^2+x_3^3-1$ and the initial set is $\mathcal{X}_0 = \{\,(x_1,x_2,x_3)^\top \mid g(\bm{x})<0\,\}$ with $g(\bm{x})= (x_1-0.3)^2+(x_2+0.2)^2 + (x_3-0.2)^2-0.001$.
\end{example}

In this example, the lower bound and upper bound of the liveness probabilities obtained through Monte Carlo methods are $\epsilon_1=0.7363$ and $\epsilon_2=0.8150$, respectively. The lower and upper bounds of the liveness probabilities computed by solving \textbf{SDP5} to \textbf{SDP7} are summarized in Table \ref{tab:vanderpol}. It is observed that as the degree $d$ of polynomials used  for computations increases, the lower and upper bounds approach the ones obtained from Monte Carlo methods.

\begin{table}[h]
    \centering
    \begin{tabular}{|c|c|c|c|c|c|c|c|}\hline 
           d &4&6&8&10&12&14&16  \\\hline
        $\epsilon_1$ by Op6 &0.5810 & 0.6861 & 0.7027 &0.7112 &0.7162 &0.7178 &0.7188\\\hline 
        $\epsilon_2$ by Op7 &0.9598 & 0.8496 & 0.8336 &0.8266 & 0.8238 & 0.8223 &0.8218\\\hline
        $\epsilon_1$ by Op5 &0.5810 & 0.6861 & 0.7027 &0.7112 &0.7162 &0.7178 &0.7188 \\\hline 
    \end{tabular}
    \caption{Computed lower and upper bounds of the liveness probabilities in Example \ref{ex:vanderpol}  ($d$ denotes the degree of unknown polynomials involved in the SDPs)}
    \label{tab:vanderpol}
\end{table}

\section{Conclusion}
\label{sec:conclusion}
This paper introduced several optimizations aimed at addressing the safe probabilistic invariance verification problem in both stochastic discrete-time and continuous-time systems. Specifically, the goal is to compute lower and upper bounds of the liveness probabilities concerning a safe set and an initial set of states. These optimizations were constructed via either using the Doob’s nonnegative supermartingale inequality-based method or relaxing the equations that characterize exact reach-avoid probabilities. The effectiveness and comparisons of these optimizations were thoroughly evaluated through both theoretical and numerical analyses.

\bibliographystyle{siamplain}
\bibliography{references}
\appendix
\section{The proof of Proposition \ref{coro:continuous}}
\label{sec:appendix_proof}
\begin{proof}
    We only need to prove that $v(\bm{x}) \geq 0$ for $\bm{x} \in \overline{\mathcal{X}}$ if $v(\bm{x})$ satisfies \eqref{eq:lower_continuous}.
    
    Assume $v(\bm{x})$ satisfies \eqref{eq:lower_continuous} and there exists $\bm{x}_0 \in \overline{\mathcal{X}}$ such that $v(\bm{x_0}) < -\delta$, where $\delta >0$. According to $\mathcal{A} v(\bm{x}) \leq 0,\forall \bm{x} \in \overline{\mathcal{X}}$, we have
    \begin{equation}
        \label{eq:coro2_1}
        \mathbb{E}[v(\widehat{\bm{X}}^{\bm{x}_0}(t,\bm{w}))] \leq v(\bm{x}_0) < -\delta, \forall t \in \mathbb{R}_{\geq 0}.
    \end{equation}

    Also since $v(\bm{x})\geq 1_{\partial\mathcal{X}}(\bm{x}) + \widehat{\mathcal{A}} u(\bm{x}), \forall \bm{x} \in \overline{\mathcal{X}}$, we have that $v(\widehat{\bm{X}}^{\bm{x}_0}(t,\bm{w})) \geq \mathcal{A}u(\widehat{\bm{X}}^{\bm{x}_0}(t,\bm{w})), \forall t \in \mathbb{R}_{\geq 0}$.
    
    Thus, we have that $\int_0^\tau \mathbb{E}[v(\widehat{\bm{X}}^{\bm{x}_0}(t,\bm{w}))]dt \geq \int_0^\tau\mathbb{E}[\mathcal{A}u(\widehat{\bm{X}}^{\bm{x}_0}(t,\bm{w}))]dt, \forall \tau \in \mathbb{R}_{\geq 0}$. 

    According to Lemma \ref{theo:dynkin} and \eqref{eq:coro2_1}, we have that
    \begin{equation*}
    \begin{split}
       \mathbb{E}[u(\widehat{\bm{X}}^{\bm{x}_0}(\tau,\bm{w}))] - u(\bm{x}_0) = & \int_0^\tau\mathbb{E}[\mathcal{A}u(\widehat{\bm{X}}^{\bm{x}_0}(t,\bm{w}))]dt \leq \int_0^\tau \mathbb{E}[v(\widehat{\bm{X}}^{\bm{x}_0}(t,\bm{w}))]dt\\
       \leq & \int_0^\tau v(\bm{x}_0)dt <  -\delta\tau,\quad \forall \tau \in \mathbb{R}_{\geq 0}, 
    \end{split}
    \end{equation*}

    implying that
    \begin{equation}
    \label{eq:coro2_3}
        \frac{\mathbb{E}[u(\widehat{\bm{X}}^{\bm{x}_0}(\tau,\bm{w}))] - u(\bm{x}_0)}{\tau} < -\delta, \forall \tau \in \mathbb{R}_{> 0}.
    \end{equation}
    
    We have that
    $\lim_{\tau\rightarrow\infty}\frac{\mathbb{E}[u(\widehat{\bm{X}}^{\bm{x}_0}(\tau,\bm{w}))] - u(\bm{x}_0)}{\tau} = 0$,
    which contradicts \eqref{eq:coro2_3}. Therefore, we can conclude that if $v(\bm{x})$ satisfies \eqref{eq:lower_continuous}, $v(\bm{x}) \geq 0$ holds for $\bm{x} \in \mathcal{X}$. Thus, constraints \eqref{eq:lower_2007} and \eqref{eq:lower_continuous} are equivalent.
\end{proof}

\section{Optimizations for Computing Lower and Upper Bounds}
\label{sec:appendix_opt}
\newline
\begin{minipage}{0.48\linewidth}
    \begin{equation*}
\begin{split}
&\textbf{Op0}~~~~~~~\max_{v(\bm{x}),\epsilon_1} \epsilon_1\\
&\text{s.t.~}\begin{cases}
  v(\bm{x})\leq 1-\epsilon_1,  \qquad\quad~~\forall \bm{x}\in \mathcal{X}_0\\
       v(\bm{x})\geq \mathbb{E}^{\infty}[v(\bm{\phi}^{\bm{x}}_{\pi}(1))], ~\forall \bm{x}\in \mathbb{R}^n\\
        v(\bm{x})\geq 1,\qquad\quad~~\forall \bm{x}\in \mathbb{R}^n\setminus \mathcal{X}\\
       v(\bm{x})\geq 0, \qquad\qquad\quad~~\forall \bm{x}\in \mathbb{R}^n\\
     \epsilon_1\geq 0
\end{cases}
\end{split}
\end{equation*}
\end{minipage}
\hfill
\begin{minipage}{0.48\linewidth}
    \begin{equation*}
\begin{split}
&\textbf{Op1}~~~~~~~\max_{v(\bm{x}),\epsilon_1} \epsilon_1\\
&\text{s.t.~}\begin{cases}
   v(\bm{x})\leq 1-\epsilon_1, \qquad\quad~~\forall \bm{x}\in \mathcal{X}_0\\
 v(\bm{x})\geq \mathbb{E}^{\infty}[v(\bm{\phi}^{\bm{x}}_{\pi}(1))], ~\forall \bm{x}\in \mathcal{X}\\
v(\bm{x})\geq 1, \qquad\quad~~~ \forall \bm{x}\in \widehat{\mathcal{X}}\setminus \mathcal{X}\\
v(\bm{x})\geq 0, \qquad\qquad\quad~~ \forall \bm{x}\in \widehat{\mathcal{X}}\\
\epsilon_1\geq 0
\end{cases}
\end{split}
\end{equation*}
\end{minipage}

\noindent
\begin{minipage}{0.43\linewidth}
    \begin{equation*}
\begin{split}
&\textbf{Op2}~~~~~~~\min_{v(\bm{x}),\epsilon_2} \epsilon_2\\
&\text{s.t.~}\begin{cases}
            v(\bm{x})\leq \epsilon_2, \qquad\quad~~\forall \bm{x}\in \mathcal{X}_0\\
            \mathbb{E}^{\infty}[v(\bm{\phi}_{\pi}^{\bm{x}}(1))]-v(\bm{x}) \leq -\delta,\\
            \qquad\qquad\qquad\qquad~\forall \bm{x} \in \mathcal{X}\\
            v(\bm{x})\geq 0, \qquad\qquad\,\forall \bm{x}\in \widehat{\mathcal{X}}\\
       \epsilon_2\geq 0
\end{cases}
\end{split}
\end{equation*}
\end{minipage}
\hfill
\begin{minipage}{0.56\linewidth}
    \begin{equation*}
\begin{split}
&\textbf{Op3}~~~~~~~\max_{v(\bm{x}),w(\bm{x}),\epsilon_1} \epsilon_1\\
&\text{s.t.~}\begin{cases}
            v(\bm{x})\leq 1-\epsilon_1, \qquad\qquad\qquad~~~ \forall \bm{x}\in \mathcal{X}_0\\
       v(\bm{x})\geq \mathbb{E}^{\infty}[v(\bm{\phi}^{\bm{x}}_{\pi}(1))],\qquad\quad~~\forall \bm{x}\in \mathcal{X}\\
       v(\bm{x})\geq \mathbb{E}^{\infty}[w(\bm{\phi}^{\bm{x}}_{\pi}(1))]-w(\bm{x}), \forall \bm{x}\in \mathcal{X}\\
       v(\bm{x})\geq 1, \qquad\qquad\qquad\quad\, \forall \bm{x}\in \widehat{\mathcal{X}}\setminus \mathcal{X}\\
       \epsilon_1\geq 0
\end{cases}
\end{split}
\end{equation*}
\end{minipage}

\noindent
\begin{minipage}{0.57\linewidth}
    \begin{equation*}
\begin{split}
&\textbf{Op4}~~~~~~~\min_{v(\bm{x}),w(\bm{x}),\epsilon_2} \epsilon_2\\
&\text{s.t.~}\begin{cases}
      v(\bm{x})\geq 1-\epsilon_2, \qquad\qquad\qquad\quad\, \forall \bm{x}\in \mathcal{X}_0\\
       v(\bm{x})\leq \mathbb{E}^{\infty}[v(\bm{\phi}^{\bm{x}}_{\pi}(1))], \qquad\qquad \forall \bm{x}\in \mathcal{X}\\
       v(\bm{x})\leq \mathbb{E}^{\infty}[w(\bm{\phi}^{\bm{x}}_{\pi}(1))]-w(\bm{x}),~ \forall \bm{x}\in \mathcal{X}\\
       v(\bm{x})\leq 1, \qquad\qquad\qquad\quad~\, \forall \bm{x}\in \widehat{\mathcal{X}}\setminus \mathcal{X}\\
       \epsilon_2\geq 0
\end{cases}
\end{split}
\end{equation*}
\end{minipage}
\hfill
\begin{minipage}{0.38\linewidth}
    \begin{equation*}
\begin{split}
&\textbf{Op5}~~~~~~~\max_{v(\bm{x}),\epsilon_1} \epsilon_1\\
&\text{s.t.~}\begin{cases}
      v(\bm{x}) \leq 1 - \epsilon_1, & \forall \bm{x} \in \mathcal{X}_0\\
       \mathcal{A} v(\bm{x}) \leq 0, &\forall \bm{x} \in \overline{\mathcal{X}} \\
       v(\bm{x}) \geq 0, & \forall \bm{x} \in \overline{\mathcal{X}} \\
       v(\bm{x})\geq 1, & \forall \bm{x} \in \partial \mathcal{X}\\
       \epsilon_1 \geq 0
\end{cases}
\end{split}
\end{equation*}
\end{minipage}

\noindent
\begin{minipage}{0.43\linewidth}
    \begin{equation*}
\begin{split}
&\textbf{Op6}~~~~~~~\max_{v(\bm{x}),u(\bm{x}),\epsilon_1} \epsilon_1\\
&\text{s.t.~}\begin{cases}
      v(\bm{x}) \leq 1 - \epsilon_1, \quad\,\forall \bm{x} \in \mathcal{X}_0 \\
     \mathcal{A} v(\bm{x}) \leq 0, \quad\quad~\,\forall \bm{x} \in \mathcal{X} \\
       v(\bm{x}) \geq \mathcal{A} u(\bm{x}), \quad\forall \bm{x} \in \mathcal{X}\\
       v(\bm{x})\geq 1, \quad\quad~~\forall \bm{x} \in \partial{\mathcal{X}}\\
       \epsilon_1 \geq 0
\end{cases}
\end{split}
\end{equation*}
\end{minipage}
\hfill
\begin{minipage}{0.43\linewidth}
    \begin{equation*}
\begin{split}
&\textbf{Op7}~~~~~~~\min_{v(\bm{x}),u(\bm{x}),\epsilon_2} \epsilon_2\\
&\text{s.t.~}\begin{cases}
      v(\bm{x}) \geq 1 - \epsilon_2, \quad\, \forall \bm{x} \in \mathcal{X}_0 \\
       \mathcal{A} v(\bm{x}) \geq 0, \quad\quad~\, \forall \bm{x} \in \mathcal{X} \\
       v(\bm{x}) \leq \mathcal{A} u(\bm{x}), \quad\forall \bm{x} \in \mathcal{X}& \\
       v(\bm{x})\leq 1, \quad\quad~~\forall \bm{x} \in \partial{\mathcal{X}}\\
       \epsilon_2\geq 0
\end{cases}
\end{split}
\end{equation*}
\end{minipage}

\section{Semi-definite Programming Implementation}
\label{sec:appendix_sdp}
\newline
\begin{minipage}{0.48\linewidth}
\begin{equation*}
\begin{split}
&\textbf{SDP0}~~~~~~~\min_{\epsilon_1,v(\bm{x}),s_i(\bm{x}),i=0,\ldots,1} -\epsilon_1\\
&\text{s.t.~}\\
&\begin{cases}
  1-\epsilon_1-v(\bm{x}) +s_0(\bm{x})g(\bm{x})\in \sum[\bm{x}],\\
       v(\bm{x})-\mathbb{E}^{\infty}[v(\bm{\phi}^{\bm{x}}_{\pi}(1))]\in \sum[\bm{x}],\\
       v(\bm{x})-1-s_1(\bm{x})h(\bm{x})\in \sum[\bm{x}],\\
       v(\bm{x})\in \sum[\bm{x}],\\
       \epsilon_1\geq 0,\\
       s_0(\bm{x})\in \sum[\bm{x}],s_1(\bm{x})\in \sum[\bm{x}].
\end{cases}
\end{split}
\end{equation*}
\end{minipage}
\hfill
\begin{minipage}{0.48\linewidth}
\begin{equation*}
\begin{split}
&\textbf{SDP1}~~~~~~~\min_{\epsilon_1,v(\bm{x}),s_i(\bm{x}),i=0,1} -\epsilon_1\\
&\text{s.t.~}\\
&\begin{cases}
  1-\epsilon_1-v(\bm{x})+s_0(\bm{x})g(\bm{x})\in \sum[\bm{x}],\\
       v(\bm{x})-\mathbb{E}^{\infty}[v(\bm{\phi}^{\bm{x}}_{\pi}(1))]\\
       ~~~~~~~~~~~~~~~~~+s_1(\bm{x})h(\bm{x})\in \sum[\bm{x}],\\
       v(\bm{x})-1+s_2(\bm{x})\widehat{h}(\bm{x})\\
       ~~~~~~~~~~~~~~~~~-s_3(\bm{x})h(\bm{x}) \in \sum[\bm{x}],\\
       v(\bm{x})+s_4(\bm{x})\widehat{h}(\bm{x})\in \sum[\bm{x}],\\
       \epsilon_1\geq 0,\\
       s_i(\bm{x})\in \sum[\bm{x}], i=0,\ldots,4.
\end{cases}
\end{split}
\end{equation*}
\end{minipage}

\noindent
\begin{minipage}{0.472\linewidth}
\begin{equation*}
\begin{split}
&\textbf{SDP2}~\min_{v(\bm{x}),\epsilon_2,s_i(\bm{x}),i=0,\ldots,2} \epsilon_2\\
&\text{s.t.~}\\
&\begin{cases}
        \epsilon_2-v(\bm{x})+s_0(\bm{x})g(\bm{x})\in \sum[\bm{x}],\\
        -\delta-\mathbb{E}^{\infty}[v(\bm{\phi}_{\pi}^{\bm{x}}(1))]+v(\bm{x})\\
        ~~~~~~~~~~~~~+s_1(\bm{x})h(\bm{x})\in \sum[\bm{x}],\\
        v(\bm{x})+s_2(\bm{x})\widehat{h}(\bm{x})\in \sum[\bm{x}],\\
   \epsilon_2\geq 0,\\
   s_i(\bm{x})\in \sum[\bm{x}], i=0,\ldots,2,
\end{cases}
\end{split}
\end{equation*}
where $\delta=10^{-6}$.
\end{minipage}
\hfill
\begin{minipage}{0.472\linewidth}
\begin{equation*}
\begin{split}
&\textbf{SDP3}~\min_{\epsilon_1,v(\bm{x}),w(\bm{x}),s_i(\bm{x}),i=0,\ldots,4} -\epsilon_1\\
&\text{s.t.~}\\
&\begin{cases}
    1-\epsilon_1-v(\bm{x})+s_0(\bm{x})g(\bm{x})\in \sum[\bm{x}],\\
    v(\bm{x})-\mathbb{E}^{\infty}[v(\bm{\phi}^{\bm{x}}_{\pi}(1))]\\
    ~~~~~~~~~~~~~~~~~+s_1(\bm{x})h(\bm{x})\in \sum[\bm{x}],\\
   v(\bm{x})-\mathbb{E}^{\infty}[w(\bm{\phi}^{\bm{x}}_{\pi}(1))]+w(\bm{x})\\
   ~~~~~~~~~~~~~~~~~+s_2(\bm{x})h(\bm{x})\in \sum[\bm{x}], \\
   v(\bm{x})-1+s_3(\bm{x})\widehat{h}(\bm{x})\\
   ~~~~~~~~~~~~~~~~~-s_4(\bm{x})h(\bm{x}) \in \sum[\bm{x}],\\
   \epsilon_1\geq 0,\\
    s_i(\bm{x})\in \sum[\bm{x}], i=0,\ldots,4.
\end{cases}
\end{split}
\end{equation*}
\end{minipage}

\noindent
\begin{minipage}{0.472\linewidth}
\begin{equation*}
\begin{split}
&\textbf{SDP4}~\min_{\epsilon_2,v(\bm{x}),w(\bm{x}),s_i(\bm{x}),i=0,\ldots,4} \epsilon_2\\
&\text{s.t.~}\\
&\begin{cases}
    v(\bm{x})-1+\epsilon_2+s_0(\bm{x})g(\bm{x})\in \sum[\bm{x}],\\
    \mathbb{E}^{\infty}[v(\bm{\phi}^{\bm{x}}_{\pi}(1))]-v(\bm{x})\\
    ~~~~~~~~~~~~~~~~~+s_1(\bm{x})h(\bm{x})\in \sum[\bm{x}],\\
    -v(\bm{x})+\mathbb{E}^{\infty}[w(\bm{\phi}^{\bm{x}}_{\pi}(1))]-w(\bm{x})\\
    ~~~~~~~~~~~~~~~~~+s_2(\bm{x})h(\bm{x})\in \sum[\bm{x}], \\
    1-v(\bm{x})+s_3(\bm{x})\widehat{h}(\bm{x})\\
    ~~~~~~~~~~~~~~~~~-s_4(\bm{x})h(\bm{x}) \in \sum[\bm{x}],\\
    \epsilon_2\geq 0,\\
    s_i(\bm{x})\in \sum[\bm{x}], i=0,\ldots,4.
\end{cases}
\end{split}
\end{equation*}
\end{minipage}
\hfill
\begin{minipage}{0.473\linewidth}
\begin{equation*}
\begin{split}
&\textbf{SDP5}~\min_{\epsilon_1,v(\bm{x}),p(\bm{x}),s_i(\bm{x}),i=0,1,2} -\epsilon_1\\
&\text{s.t.~}\\
&\begin{cases}
  1-\epsilon_1-v(\bm{x})+s_0(\bm{x})g(\bm{x})\in \sum[\bm{x}],\\
  -\mathcal{A} v(\bm{x})+s_1(\bm{x})h(\bm{x})\in \sum[\bm{x}],\\
  v(\bm{x})+s_2(\bm{x})h(\bm{x}) \in \sum[\bm{x}],\\
  v(\bm{x})-1+p(\bm{x})h(\bm{x})\in \sum[\bm{x}],\\
  \epsilon_1\geq 0,\\
  s_i(\bm{x})\in \sum[\bm{x}], i=0,1,2.
\end{cases}
\end{split}
\end{equation*}
\end{minipage}

\noindent
\begin{minipage}{0.473\linewidth}
\begin{equation*}
\begin{split}
&\textbf{SDP6}~\min_{\epsilon_1, v(\bm{x}),u(\bm{x}),p(\bm{x}), s_i(\bm{x}),i=0,1,2} -\epsilon_1\\
&\text{s.t.~}\\
&\begin{cases}
    1-\epsilon_1-v(\bm{x})+s_0(\bm{x})g(\bm{x})\in \sum[\bm{x}],\\
    -\mathcal{A} v(\bm{x})+s_1(\bm{x})h(\bm{x})\in \sum[\bm{x}],\\
    v(\bm{x})-\mathcal{A} u(\bm{x})+s_2(\bm{x})h(\bm{x})\in \sum[\bm{x}], \\
    v(\bm{x})-1+p(\bm{x})h(\bm{x})\in \sum[\bm{x}],\\
    \epsilon_1\geq 0,\\
    s_i(\bm{x})\in \sum[\bm{x}], i=0,1,2.
\end{cases}
\end{split}
\end{equation*}
\end{minipage}
\hfill
\begin{minipage}{0.472\linewidth}
\begin{equation*}
\begin{split}
&\textbf{SDP7}~\min_{\epsilon_2,v(\bm{x}),u(\bm{x}),p(\bm{x}),s_i(\bm{x}),i=0,1,2} \epsilon_2\\
&\text{s.t.~}\\
&\begin{cases}
    v(\bm{x})-1+\epsilon_2+s_0(\bm{x})g(\bm{x})\in \sum[\bm{x}],\\
    \mathcal{A} v(\bm{x})+s_1(\bm{x})h(\bm{x})\in \sum[\bm{x}],\\
    \mathcal{A} u(\bm{x})-v(\bm{x})+s_2(\bm{x})h(\bm{x})\in \sum[\bm{x}], \\
    1-v(\bm{x})+p(\bm{x})h(\bm{x}) \in \sum[\bm{x}],\\
    \epsilon_2\geq 0,\\
    s_i(\bm{x})\in \sum[\bm{x}], i=0,1,2.
\end{cases}
\end{split}
\end{equation*}
\end{minipage}

\end{document}